\documentclass[a4paper,fleqn]{cas-sc}

\usepackage[numbers]{natbib}
\usepackage{inputenc}
\usepackage{amsmath,amssymb,physics,cancel,graphicx, subcaption, placeins}
\usepackage{nccmath}
\usepackage[english]{babel}
\usepackage{amsthm}
\graphicspath{ {./images/} }
\usepackage{atbegshi}
\usepackage{enumitem}
\usepackage{soul, hanging}
\usepackage{ulem}
\usepackage[toc,page]{appendix}

\newtheorem{lemma}{Lemma}

\def\tsc#1{\csdef{#1}{\textsc{\lowercase{#1}}\xspace}}
\tsc{WGM}
\tsc{QE}

\begin{document}
\let\WriteBookmarks\relax
\def\floatpagepagefraction{1}
\def\textpagefraction{.001}

\shorttitle{Market Power and Distributed Solar Integration in Microgrids under Limited Regulation}    

\shortauthors{Elsa Bou Gebrael, Majd Olleik, Sebastian Zwickl-Bernhard}  

\title [mode = title]{Market Power and Distributed Solar Integration in Microgrids under Limited Regulation}  



%

\author[1]{Elsa {Bou Gebrael}}[orcid=0009-0005-2289-4417]
\author[1]{Majd Olleik}[orcid=0009-0006-4082-4102]

\cormark[1]





\affiliation[1]{organization={Department of Industrial Engineering and Management, American University of Beirut},
            city={Beirut},
            country={Lebanon}}

\author[2, 3]{Sebastian Zwickl-Bernhard}[orcid=0000-0002-8599-6278]





\affiliation[2]{organization={Energy Economics Group (EEG), Technische Universität Wien},
            city={Wien},
            country={Austria}}
            
\affiliation[3]{organization={Department of Industrial Economics and Technology Management, The Norwegian University of Science and Technology},
            city={Trondheim},
            country={Norway}}

\cortext[1]{Corresponding author}



\begin{abstract}
Decentralized electricity systems increasingly emerge where centralized grids fail to provide reliable supply. In such settings, privately operated neighborhood microgrids, often based on diesel generators, exhibit significant market power, limited regulatory oversight, and high environmental externalities. In parallel, households increasingly deploy off-grid solar photovoltaic (PV) systems to gain control over electricity supply. However, these systems suffer from curtailed excess generation during peak solar hours and unreliable access at other times. While prior studies have optimized microgrids in low-reliability grid contexts from a techno-economic perspective, they largely neglect the market power exerted by monopolistic private generators. This paper addresses this gap by developing a bi-level game-theoretic model that enables household-generated electricity to be fed into the microgrid while explicitly accounting for the market power of a neighborhood diesel generator company (DGC). The regulator sets price and feed-in-tariff caps to maximize household economic surplus (HES), while the DGC acts as a profit-maximizing agent controlling access and supply, through diesel or solar PV-based generation, batteries, and the integration of household PV-owner excess generation. The model is illustrated using high-resolution empirical data from Lebanon. Results show that: (i) price and feed-in-tariff caps substantially increase HES and consistently induce significant household PV feed-in to the microgrid; (ii) higher DGC budgets or greater PV-owner penetration lead to pronounced gains in HES; and (iii) the renewable energy share reaches 60\% under base conditions and approaches 100\% at sufficiently high budgets or PV-owner penetration levels, compared to 0\% under the status quo.
\end{abstract}



\begin{keywords}
Microgrids \sep renewable energy \sep market power \sep bi-level model
\end{keywords}

\maketitle

\section{Introduction}\label{sec 1}

Reliable electricity supply remains a major challenge in many regions where centralized grids fail to meet demand due to infrastructure degradation, financial constraints, or institutional weaknesses \cite{berha2026unreliable}. In such contexts, electricity provision increasingly relies on decentralized solutions developed by consumers and private operators. These include household-level solar home systems (SHS) and neighborhood diesel generators, which often operate formally or informally in parallel with the national grid \cite{lawrie2025friend, IEA2025}. While these arrangements improve access to electricity, they can lead to lower energy efficiency, higher generation costs, and significant environmental externalities \cite{lawrie2025friend}. 

Neighborhood diesel generators have become a widespread fallback solution in regions where national grids fail to provide reliable electricity supply. Lebanon, Iraq, and Nigeria are notable examples, where geopolitical instability, sanctions, and systemic corruption have resulted in prolonged main grid outages of up to 20 hours per day \cite{AFP2021Lebanon, al2021local, obar2022navigating}. Similar reliance on diesel generators is observed in Yemen, Afghanistan, Kenya, and other Sub-Saharan countries \cite{lawrie2025friend}. When serving a neighborhood or community, these diesel generator companies (DGCs) create microgrids that operate in parallel to the national grid \cite{lawrie2025friend, al2021local}. DGCs often benefit from a natural monopoly, assuming both generator and distributor roles, which allows them to control access and maximize profits \cite{lawrie2025friend, al2021local, zwickl2025market}. Electricity prices in such microgrids reflect both the high generation cost of diesel and the exercise of market power, while environmental externalities remain largely unaccounted for. In the absence of a reliable national grid, governments have little leverage to regulate these microgrids, which effectively compensate for broader systemic shortcomings \cite{abi2018energy}.

To counter the monopoly power of DGCs and enhance household autonomy in electricity generation, many consumers are increasingly investing in off-grid rooftop photovoltaic (PV) systems \cite{IEA2025}.  
Although rooftop PV enables cleaner electricity generation and is widely regarded as sustainable ``last-mile'' solution \cite{long2022enhancing}, its intermittent nature and ad hoc deployment often result in excess generation during periods of high solar availability and insufficient supply at other times introducing inefficiencies \cite{rad2023excess, fuchs2023swarm}. Battery storage can reduce PV curtailment; however, eliminating it entirely would require substantially oversized batteries that remain underutilized for much of their lifetime, making such solutions prohibitively expensive \cite{wang2023site}.

By allowing households to trade excess PV generation with the DGC rather than curtailing it, hybrid microgrids can increase renewable energy utilization while reducing reliance on diesel generation \cite{IEA2025, sheridan2023swarm}. However, in many developing countries these microgrids are privately owned and operated by DGCs that possess substantial market power and face only limited regulatory oversight \cite{zebra2021review, valta2018comparison}. Consequently, designing policies that increase renewable energy utilization while preserving reliable electricity supply requires accounting for the strategic interaction between regulators and dominant microgrid operators.

Therefore, the purpose of this paper is to answer the following research question: 
\textit{How can a regulator with limited oversight, and unable to challenge the profitability of an entrenched DGC, design policies that promote cleaner microgrids? Specifically, such policies should increase the utilization of household PV assets, reduce reliance on diesel generation, limit unmet demand, and improve electricity affordability.}

To answer this research question, we develop a bi-level game-theoretical model in which the regulator, constrained by the entrenched profitability of the DGC, acts first, followed by the DGC’s response. 
At the first level, the regulator maximizes household economic surplus by setting an electricity price cap and a feed-in tariff floor while ensuring that the DGC maintains its current profitability. At the second level, the DGC jointly optimizes its capacity expansion decisions (in diesel generation, solar PV, and battery storage) and system operation, including electricity trading with household PV owners, to maximize profit.

The rest of this paper proceeds as follows: Section \ref{sec:related_works} reviews the literature on microgrids and FiT policies in grid-constrained settings; it also defines the research gap and highlights the contributions. Section \ref{sec:methods} details the proposed bi-level game theoretical framework along with the solution approach. Section \ref{sec:case study} presents our case study and the data used. Section \ref{sec: results} discusses the results while Section \ref{sec:conclusion} summarizes the main findings.

\section{Related works} \label{sec:related_works}

Research on electricity systems in low-reliability grid contexts has primarily evolved along two complementary streams. The first focuses on the techno-economic evaluation of swarm grids and hybrid microgrids, where investments emerge through the organic addition of generation capacity by individual participants rather than through centralized planning \cite{sheridan2023swarm}. In these systems, energy management is often treated as a secondary consideration, unlike conventionally planned microgrids that are designed to efficiently satisfy forecasted demand. The second stream examines feed-in tariff (FiT) policies as mechanisms for promoting rooftop PV adoption and renewable energy trading with the national grid. However, the application of such policies to privately operated microgrids has received comparatively little attention.

\subsection{Swarm grids and microgrids in low-reliability grid contexts}
With microgrids and swarm electrification emerging as a natural consumer response in crisis settings, regions with unreliable national grids have observed the development of renewable energy communities, with household ownership and participation being a key success factor in implementation \cite{kirchhoff2016developing, kirchhoff2019key}.
The existing literature has thoroughly discussed the techno-economic aspects of swarm grids and microgrids, focusing on the reduction of rooftop PV systems generation inefficiencies and the optimal design and pricing under multiple objectives \cite{sheridan2023swarm, chedid2020optimal, kharrich2020design, sawwas2021pool, svolba2025renewable, manas2023critical}. For example, in the Philippines, a case study shows improvements in microgrid efficiency and blackout reduction when connecting the SHS of the different households together \cite{prevedello2021benefits}. Similarly, a Yemeni case study compares three setups, including standalone private systems, a swarm grid, and a formally organized microgrid, showing that when PV-owners are allowed to participate in either the swarm grid or the microgrid, electricity costs were reduced and demand realization was improved \cite{hoffmann2019simulating}. The results of a study from Madagascar further show that the trading of excess energy between individual PV-owners (prosumers) increased consumer economic surplus through a more efficient market-clearing approach relying on a peer-to-peer algorithm, leading to lower electricity prices \cite{bertram2024local}. Applications of game theory in community grids highlight as well the benefit of peer-to-peer trading, citing better utilization of renewable energy, reduced carbon emissions, and increased social welfare, under the assumption of a neutral grid operator \cite{luo2022distributed, fernandez2021bi}.

In practice, however, the prevalence of diesel-based microgrids in fragile contexts restricts PV-owner participation \cite{lawrie2025friend}. In such microgrids, regulatory oversight is limited, and the DGC owns the grid, obstructing the adequate policy design and fair tariff setting needed for effective renewable integration \cite{zebra2021review, valta2018comparison}. In fragile states, authorities often lack both the capacity and the mandate to regulate or formalize existing microgrids \cite{sesan2024exploring, bhattacharyya2018mini}, giving rise to market conditions that differ substantially from the competitive settings typically assumed in the literature.

\subsection{FiT policies}
The general consensus around the role of feed-in-tariff policies at the consumer level is that they incentivize investments in renewable sources \cite{choi2018prices}, and enhance the efficiency of already-installed assets \cite{al2021utilization}. Developed countries have led in establishing policies and mechanisms to manage prosumer energy exchange with the grid, including at the microgrid level, with particular consideration for greenhouse emission reduction \cite{feleafel2025feed}. Extensive frameworks deriving the optimal price of electricity fed have been developed. For example, a sequential model-based optimization is proposed for flexible FiT design in microgrids, based on historical and projected data, while keeping in line with Australian regulations \cite{habib2025smart}.

When implementing FiT policies, lower income countries aim to tackle additional concerns, such as electricity access, availability and affordability. The connection of rooftop PV systems to the grid often corresponds to the cost-optimal configuration, and PV-owners are generally willing to participate in FiT programs \cite{suberi2024rooftop}. It is therefore not surprising that multiple developing countries have set national FiT policies, with considerations for the availability of renewable energy technologies \cite{bhattacharyya2013regulate}. Hence, specific frameworks for the financial modeling of FiTs have been developed in accordance with local market complexities, and country-specific applications, such as in Kenya and Malawi \cite{moner2016adaptation}. However, such national policies are not adequate when national grids are intermittent, as feed-in can only happen when the grid is energized.
Accounting for this limitation, a decision analysis framework has been proposed to plan hybrid renewable energy systems under uncertain grid interconnection in Lebanon, explicitly considering the unreliability of the country's national utility in the discussion of future FiTs \cite{olleik2026planning}.

Tanzania is one of the only countries where considerations for local challenges and microgrid-level FiT policies overlap. It explores how FiTs can be adapted specifically for remote mini-grids, determining a tariff reflecting rural operating conditions, and competing with diesel generator economics \cite{bhattacharyya2013regulate, moner2016adaptation}.

\subsection{Research gap and contributions}

Overall, the existing literature demonstrates the potential of hybrid microgrids to improve renewable energy utilization, enhance electricity reliability, and reduce electricity costs. Likewise, a growing body of research has investigated mechanisms for utilizing excess generation from household-owned rooftop PV systems, primarily through national FiT schemes. However, such mechanisms have received little attention in privately operated microgrids, whose emergence is itself a consequence of limited regulatory capacity \cite{sesan2024exploring}. Moreover, most studies on hybrid PV-diesel microgrids and local energy trading assume competitive or otherwise simplified market structures, overlooking the presence of DGCs and their strategic interactions with regulators and consumers. Consequently, the role of market power and institutional constraints in shaping renewable energy integration remains largely unexplored \cite{zwickl2025market, svolba2025renewable, sesan2024exploring}.

To address this gap, this paper makes the following contributions:

\begin{itemize}

\item We develop a bi-level game-theoretic framework that explicitly models the strategic interaction between a weak regulator and a DGC operating a crisis-driven microgrid. Unlike existing approaches that assume centralized planning or competitive market structures, the proposed framework captures the institutional reality in which the regulator can only influence market outcomes indirectly through an electricity price cap and a feed-in tariff floor while preserving the DGC's current profitability.

\item We jointly model regulatory policy design, generation capacity expansion, and microgrid operation within a unified optimization framework. In particular, the DGC simultaneously determines investments in diesel generation, solar PV, and battery storage, together with electricity prices, feed-in tariffs, and operational decisions, allowing the assessment of trade-offs between DGC-owned renewable capacity and electricity purchases from household PV owners. While the case study focuses on diesel, solar PV, and battery storage, the proposed framework can readily accommodate other technologies, including wind power, enabling its application to a broader range of microgrid settings.

\item We calibrate the proposed framework using field data from a Lebanese crisis-driven microgrid, including high-resolution household demand profiles collected through dedicated logging devices. The resulting case study provides regulatory and planning insights into the design of electricity price caps and feed-in tariffs under weak regulatory oversight, with implications for Lebanon and other countries relying on crisis-driven microgrids.

\end{itemize}

\section{Modeling framework}\label{sec:methods}

The proposed modeling framework considers that the regulatory entity, acting as a leader in a bi-level game, is interested in achieving the following four goals: (i) increasing the utilization of the existing household PV assets, (ii) reducing reliance on diesel-based generation, (iii) limiting the unmet demand, and (iv) reducing the electricity price in the DGC-controlled microgrid. This should be done without reducing the profits of the DGC. Achieving these desired regulatory goals contributes to increasing the household economic surplus composed of the combined value of (i) the household satisfied demand, and (ii) the household PV electricity fed into the microgrid. The regulatory entity is able to set an upper bound on the price of electricity sold to the households $P^{max}$, and a lower bound on the price of electricity sold to the DGC by the household PV-owners $FiT^{min}$ (Figure \ref{fig:BLG diagram}).

\begin{figure}
    \centering
    \includegraphics[width=0.9\linewidth]{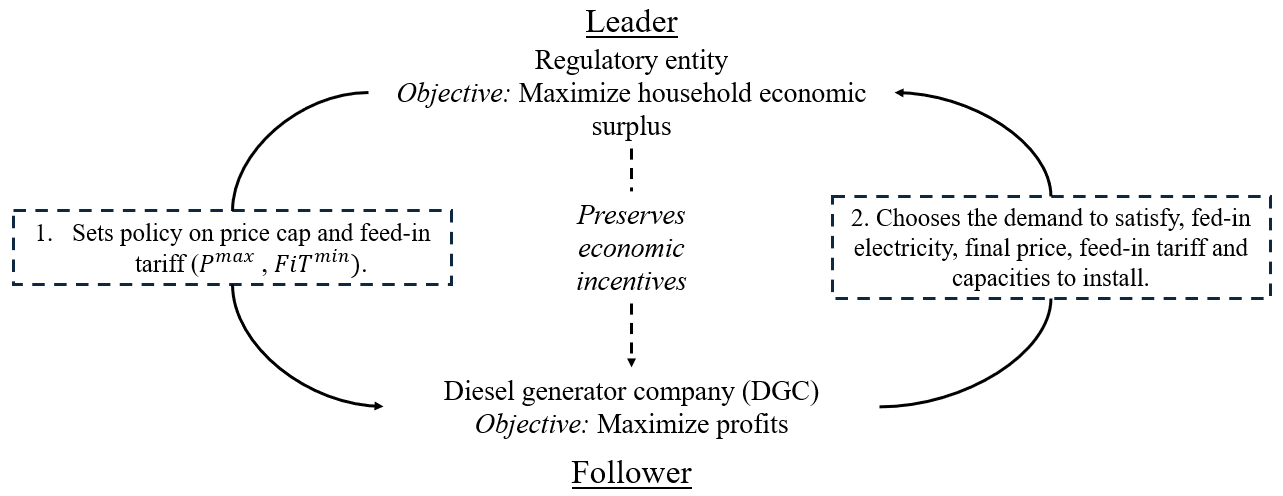}
    \caption{Diagram of the bi-level game}
    \label{fig:BLG diagram}
\end{figure}

The DGC acts as a follower interested in maximizing its profits in response to the policy adopted by the regulatory entity composed of the tuple $(P^{max}, FiT^{min})$. Traditionally, drawing on practical experiences in developing countries with diesel-based microgrids, the DGC controls the generation schedule (and accordingly the outage schedule) and the household access to the microgrid \cite{lawrie2025friend} (Figure \ref{fig:MG_statusquo}). It sells the electricity at a price $P$ of its choice such that $P \leq P^{max}$. In our proposed framework, the DGC is offered the additional flexibility of purchasing electricity from households with PV systems at a price $FiT$ of its choice as long as $FiT \geq FiT^{min}$. The DGC is also free to choose to invest in its own PV and battery storage assets if such an investment is economically attractive (Figure \ref{fig:MG_new}). Formally speaking, in this Stackelberg game, the leader sets a policy affecting the feasible set of the follower, and evaluates its own payoff by anticipating the follower's reaction.

\begin{figure}
    \centering
    \begin{subfigure}[b]{0.48\textwidth}
        \centering
        \includegraphics[width=\textwidth]{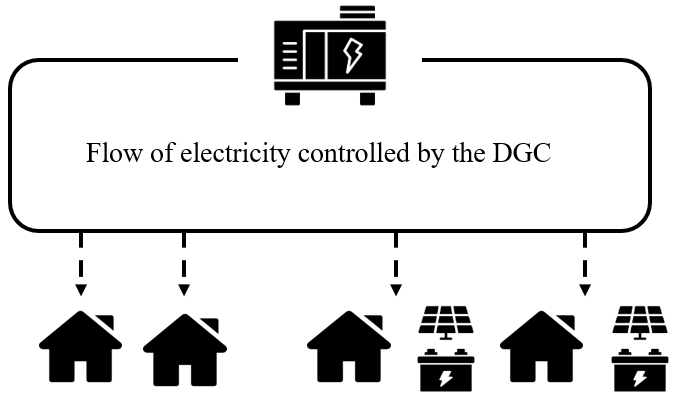}
        \caption{Status-quo microgrid}
        \label{fig:MG_statusquo}
    \end{subfigure}
    \hfill
    \begin{subfigure}[b]{0.48\textwidth}
        \centering
        \includegraphics[width=\textwidth]{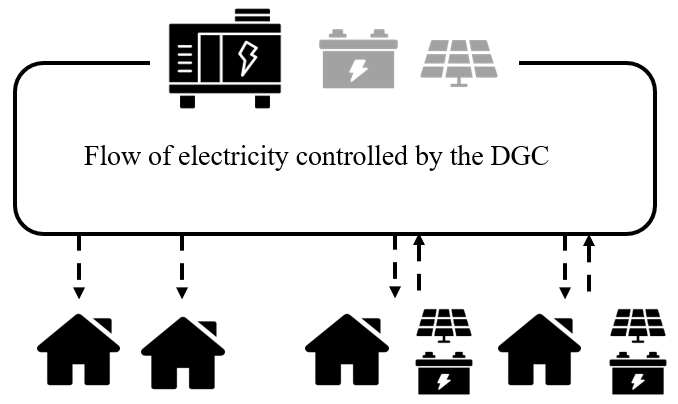}
        \caption{Proposed microgrid}
        \label{fig:MG_new}
    \end{subfigure}
    \begin{subfigure}[b]{0.75\textwidth}
        \centering
        \includegraphics[width=\textwidth]{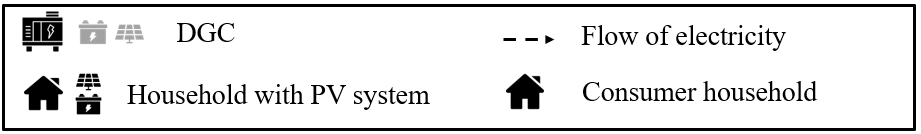}
    \end{subfigure}
    \caption{Representative diagram of the microgrid before and after change}
    \label{fig:microgrids}
\end{figure}

Accordingly, the full game theoretical model can be written as:
\begin{align*}
    &\max_{P^{max}, FiT^{min}} \quad \text{household economic surplus} \tag{leader obj} \label{L obj}\\
    &\mbox{\textit{subject to:}}\\
    &\quad\quad\quad\quad\text{ensuring: } \text{profits}^{DGC} \geq \text{base profits}^{DGC}\tag{leader constraint}\label{RE cstrt}\\
    &\mbox{\textit{and:}}\\
    &\quad\quad\quad\quad\max_\mathcal{S}\quad\text{profits}^{DGC}
    \tag{follower obj}\label{F obj}\\
    &\quad\quad\quad\quad\mbox{\textit{subject to:}}\\
    &\quad\quad\quad\quad\quad\quad\quad\quad\mbox{economic and technical constraints}\tag{follower constraints}\label{DGC cstrt}
\end{align*}

The first level decisions are composed of the tuple ($P^{max}, FiT^{min})$. Table \ref{tab:DV} presents the decision variables used in the second-level model, which constitute set $\mathcal{S}$. The parameters and relevant set names are shown in Appendix \ref{nomenclature}, in Tables \ref{tab:In} and \ref{tab:sets} respectively. To preserve the tractability of the model, we employ representative days, each weighed by $\omega_d$, characterizing the entire year. For all variables and parameters, the indices $i,g,y,d,h$ refer to the household type (non PV-owner: 0, PV-owner: 1), the generation technology, the year, the representative day, and the hour, respectively. Moreover, all decision variables are nonnegative, and represented by Latin letters, while parameters are represented by Greek ones.

The two levels of the game are further detailed in Sections \ref{sec: FL} and \ref{sec: SL}.
\begin{table}[]
    \centering
    \small
    \caption{Second-level (DGC) decision variables}
    \begin{tabular}{l p{300pt} l}
        \hline
        Symbol & Variable name & Unit\\
        \hline
         $A_{g,y}$ & Added DGC capacity & kW\\
         $B^+_{y,d,h}$ & Battery charge & kW\\
         $B^-_{y,d,h}$ & Battery discharge & kW\\
         $b_{k,y,d,h}$ & Binary variable enforcing the $k^{th}$ constraint on the heat rate curve & binary \\
         $C_{g,y}$ & Capacity installed & kW\\
         $D_{g,y,d,h}$ & Dispatched power  & kW \\
         $Fi_{i,y,d,h}$ & Fed-in electricity from household PV-owner excess & kW\\
         $FiT$ & Feed-in-tariff & USD/kWh\\
         $P$ & Price of electricity & USD/kWh\\
         $Q_y$ & Demand served by the microgrid & kWh\\
         $Ret_{g,y}$ & Retired capacity & kW\\
         $SoC_{y,d,h}$ & State of charge at the end of every hour $h$ & kWh\\
         $SoC^0_{y,d}$ & State of charge at the beginning of the day & kWh\\
         $U_{y,d,h}$ & Unmet demand across all households & kWh\\
         $R_{y,d,h}$ & Diesel consumption as a function of the diesel generator heat rate and the dispatch & L/kWh\\ 
        $S_g$ & Salvage value of the installed capacity at the end of the planning horizon & USD\\
         $TR_y$ & Total revenues earned by the DGC  & USD\\
         $TCC_y$ & Total capital costs incurred by the DGC & USD\\
         $TOVC_y$ & Total operation variable costs incurred by the DGC & USD\\
         $TOFC_y$ & Total operation fixed costs incurred by the DGC& USD\\
         $TUDC_y$ & Total unmet demand costs incurred by the DGC& USD\\
         \hline
        \end{tabular}
        
        \label{tab:DV}
\end{table} 

\subsection{First-level} \label{sec: FL}
The regulatory entity's payoff, equivalent to the discounted household economic surplus (HES), is given by Equation \ref{RE obj}: 
\begin{equation}
    \max_{P^{max}, FiT^{min}} \sum_{y\in\mathcal{Y}} \left[Q_y \left(\xi - P\right) + \sum_{d\in\mathcal{D}} \omega_d \sum_{h\in\mathcal{H}} Fi_{i,y,d,h} \times FiT\right] \left((1+\gamma^{RE})^{-y}\right) \label{RE obj}
\end{equation}
It is made up of two components: (i) the value of the demand met, which can be computed as the product of the met demand $Q_y$ (detailed in Equation \ref{DS}) and the difference between the value of lost load $\xi$ and the price $P$, and (ii) the value from feeding in household-generated electricity to the microgrid, computed as the product of the fed-in electricity $Fi_{y,d,h}$ and the feed-in-tariff $FiT$. The symbol $\gamma^{RE}$ denotes the discount rate of the regulatory entity.

The total demand in the microgrid is expressed as $- \sum_{i\in\mathcal{I}} \min\left( 0,  \Omega_{i,y} \times \sigma_{i,y,d,h}\right)$ where the symbols $\Omega_{i=0, y}$ and $\Omega_{i=1, y}$ respectively represent the number of non PV-owner households and PV-owner households in the microgrid. The symbol $\sigma_{i,y,d,h}$ denotes the excess electricity supply for household category $i$. A positive $\sigma_{i,y,d,h}$ means that household generation potential from PV exceeds its instantaneous demand, and that PV-owners can feed at most $\sigma_{i,y,d,h}$ into the microgrid. A negative $\sigma_{i,y,d,h}$ indicates a net demand.

The served demand $Q_y$ is defined in Equation \ref{DS}. It is derived from the total demand in the microgrid, reduced by the unserved demand $U_{y,d,h}$ determined by the DGC:
\begin{equation} \label{DS}
   Q_y = \sum_{d\in\mathcal{D}} \omega_d \sum_{h\in\mathcal{H}} \left( - \sum_{i\in\mathcal{I}} \min\left( 0,  \Omega_{i,y} \times \sigma_{i,y,d,h}\right) - U_{y,d,h}\right)\quad\forall y
\end{equation}
The PV-owner's generation excess is the difference between the household PV total generation and the household demand, whereas the non-PV-owner's excess is simply the negative demand, as defined below:
\begin{align*}
    \sigma_{i=0,y,d,h}&=-\mu_{i=0, y, d, h}\quad \forall y,d,h,i\\
    \sigma_{i=1,y,d,h} &= \phi_{y,d,h}\times \theta^{PV}_{i=1} - \mu_{i,y,d,h} \quad \forall y,d,h,i
\end{align*}
where $\mu_{i, y, d, h}$ is the household demand, $\phi_{y,d,h}$ is the solar capacity factor, and $\theta^{PV}_{i=1}$ is the average size of the household PV-owner PV system.

The decision variables of the regulator are only bounded by 0 on the lower end. Any nonnegative feed-in-tariff is acceptable to PV-owners, as their excess generation potential would have been otherwise wasted. However, $(P^{max}, FiT^{min})$ must keep the DGC economically incentivized, \textit{i.e.,} the maximized net present value (NPV) of profits of the DGC under any policy ($NPV^{DGC}$) should be greater than or equal to the one under the current market conditions ($NPV_0^{DGC}$) (Equation \ref{RE cstrt}), establishing the constraint for the DGC's participation in the game.
The DGC's profit maximizing model, presented in Section \ref{sec: SL}, is implicit to this constraint.

\subsection{Second-level}\label{sec: SL}
The below model describing the microgrid under DGC control is inspired by the available literature around determining the optimal generation asset sizing in a microgrid under a profit-maximization objective \cite{villa2022oversizing, karimi2024multi}. The DGC desires to maximize its NPV. We assume that, if multiple sets of decisions lead to the same maximized NPV, the DGC prefers the one with the lowest level of unmet demand (Equation \ref{eq: DGC obj}). This preference is illustrated by a very small unmet demand penalty weight $\varepsilon$ introduced to the DGC's objective function:
\begin{equation} \label{eq: DGC obj}
    \max_{\mathcal{S}} \quad NPV^{DGC} - \varepsilon\left(\sum_y \sum_d\omega_d\sum_h U_{y, d ,h}\cdot(1+\gamma^{DGC})^{-y}\right)
\end{equation}

The NPV formulation is detailed in Equation \ref{obj fun}:
\begin{equation} \label{obj fun}
    NPV^{DGC} = \sum_{y\in\mathcal{Y}} \left[\left(TR_y-TCC_y- TOVC_y-TOFC_y\right)(1+\gamma^{DGC})^{-y}\right] + \sum_{g\in\mathcal{G}} S_g (1+\gamma^{DGC})^{-\Upsilon}
\end{equation}

For each year $y$, $TR_y$ is the total revenues, $TCC_y$ the total capital costs, $TOVC_y$ the total operation variable costs, $TOFC_y$ the total operation fixed costs, and $S_g$ the salvage value of the technology $g$. Each of these terms are detailed in Equations \ref{eq:TCC} through \ref{eq:TR} $\forall y$, and \ref{eq:SV} $\forall g$, where the summations indexed by $y$, $d$, $h$, $g$, $i$ are over the entire sets $\mathcal{Y}$, $\mathcal{D}$, $\mathcal{H}$, $\mathcal{G}$, $\mathcal{I}$, unless otherwise specified.
\begin{align}
    TCC_y&= \sum_{g} \left(\lambda^C_{g,y} \cdot A_{g,y}\right)\quad\forall y\label{eq:TCC}\\
    TOVC_y &= \sum_d \omega_d \sum_h \left( \lambda^{OV}_{g=B} \cdot B^+_{y,d,h} + \sum_{g\in\mathcal{G}g} \left(\lambda^{OV}_g \cdot D_{g,y,d,h}\right) + R_{y,d,h}\cdot \pi + \sum_i Fi_{i,y,d,h} \cdot FiT\right)\quad\forall y\label{eq:TOVC}\\
    TOFC_y &= \sum_{g} \left( \lambda^{OF}_g \cdot C_{g,y} \right)\quad\forall y\label{eq:TOFC}\\
    TR_y&= P\sum_d \omega_d \sum_h \left(B^-_{y,d,h} - B^+_{y,d,h} + \sum_{g \in \mathcal{G}g}D_{g,y,d,h} + \sum_i Fi_{i,y,d,h}\right)\quad\forall y \label{eq:TR}\\
    S_g &= \sum_{y = \Upsilon - \nu_g} ^ {\Upsilon-1} 
    \left[ A_{g,y} \cdot \lambda^C_{g, \Upsilon} \cdot \left(1- \frac{\Upsilon - y}{\nu_g} \right) \right] \quad\forall g\label{eq:SV}
\end{align}

The change from the status quo to the proposed microgrid is represented by the binary input $\beta$, with $\beta = 1$ indicating, from the perspective of the DGC, the possibility of feeding in electricity from households and installing PV and battery capacities. Equations \ref{max PV} and \ref{max_bat} limit the installation of PV and batteries respectively to 0 in the current situation, or to a large enough number $M$ otherwise. Equation \ref{fed-in} sets an upper bound on the PV-owners' feed-in when $\beta = 1$: 

\begin{align} 
    C_{g=PV, y} &\leq \beta \times M \quad\forall y \label{max PV}\\
    C_{g=B, y} &\leq \beta \times M \quad\forall y \label{max_bat}\\
    Fi_{i,y,d,h} &\leq \max\left(0, \Omega_{i,y} \times \sigma_{i,y,d,h}\right) \times \beta \quad \forall y,d,h,i \label{fed-in}
\end{align}

The DGC's objective function is subject to the following constraints:
\begin{itemize}
    \item The supply-demand balance constraint:\\
    \begin{equation} \label{sup-dem}
    U_{y,d,h} + B^-_{y,d,h} + \sum_{g\in\mathcal{G}g} D_{g,y,d,h} + \sum_i Fi_{i,y,d,h}  = B^+_{y,d,h} - \sum_i \min\left( 0,  \Omega_{i,y} \times \sigma_{i,y,d,h}\right)\quad\forall y,d,h
    \end{equation}
    Equation \ref{sup-dem} ensures that the dispatched energy $D_{g,y,d,h}$, the discharging of batteries $B^-_{y,d,h}$ and the fed-in electricity $Fi$ are enough to satisfy the household demand $ - \sum_i \min\left( 0,  \Omega_{i,y} \times \sigma_{i,y,d,h}\right)$ and the charging of batteries $B^+_{y,d,h}$, while allowing unmet demand $U_{y, d, h}$.
    Furthermore, at any hour, the unmet demand must not exceed the total demand (Equation \ref{min UD h}):
    \begin{equation}\label{min UD h}
        U_{y, d, h} \leq - \sum_i \min\left(0,  \Omega_{i,y} \times \sigma_{i,y,d,h}\right)
        \quad \quad\forall y, d, h
    \end{equation}

    \item The budget constraint:\\
    While the DGC can choose which capacities to install, its decision is bounded by a total budget $\Pi$. This budget is available at year $y=0$, and all subsequent discounted cash flows due to capacity installation expenses should amount to it, as expressed in Equation \ref{eq: budget}.
    \begin{equation} \label{eq: budget}
        \sum_{g\in\mathcal{G}g}\sum_{y} A_{g,y}\times \lambda^C_{g, y}\times(1+\gamma^{DGC})^{-y} \leq \Pi
    \end{equation}

    \item The regulator policy constraints:\\
    The electricity price and the FiT set by the DGC are constrained by the regulator's decisions on $P^{max}$ and $FiT^{max}$ (Equations \ref{p bound} and \ref{fit bound}).
        \begin{align}
            P&\leq P^{max} \label{p bound}\\
            FiT&\geq FiT^{min}\label{fit bound}
        \end{align}

\item The technologies' capacity and retirement constraints:\\
Appendix \ref{gen tech} details the constraints on installation and retirement of capacity for the different technologies.

\item The dispatch constraints:\\
The dispatch (or in the case of batteries, charge and discharge) for each technology is limited by its corresponding installed capacity, and a corresponding capacity factor when applicable. The equations for these constraints are in appendix \ref{gen disp}. The diesel generator should also obey a set of constraints pertaining to its heat rate. Traditionally, the heat rate function is a U-shaped curve showing the fuel consumed per generated kWh against the utilization rate of a diesel generator. To preserve the linearity of constraints, we approximate the heat rate curve by a piecewise constant function \cite{wuijts2024effect}. In the following, we model the curve as three piecewise constant functions where the diesel consumption $R_{y,d,h}$ is defined as:
    \begin{align}
        R_{y,d,h} &= \rho_1 \times D_{g=DG, y,d,h}\quad \text{if} \quad D_{g=DG, y,d,h} \leq 0.30 \times C_{g=DG,y} \label{HR.1}\tag{HR.1}\\
        R_{y,d,h} &= \rho_2 \times D_{g=DG, y,d,h}\quad \text{if} \quad\begin{cases} 0.30 \times C_{g=DG,y} < D_{g=DG, y,d,h}\\
        D_{g=DG, y,d,h} \leq 0.60 \times C_{g=DG,y}\end{cases}\label{HR.2}\tag{HR.2}\\
        R_{y,d,h} &= \rho_3 \times D_{g=DG, y,d,h}\quad \text{if} \quad 0.60 \times C_{g=DG,y} < D_{g=DG, y,d,h} \label{HR.3}\tag{HR.3}
    \end{align}
    $\forall y,d,h$. \\

\item The storage technologies constraints:\\
Appendix \ref{storage tech} presents the constraints related to the tracking of the batteries state of charge.
\end{itemize}

\subsection{Solution approach}
The second-level model proposed in this paper is non-linear in the equations defining $TOVC_y$ and $TR_y$ (Equations \ref{eq:TOVC} and \ref{eq:TR} respectively). The non-linearity is caused by the decision variables $P$ and $FiT$. However, lemma \ref{GD} shows that $P$ and $FiT$ are optimized at $P^{max}$ and $FiT^{min}$ respectively, and therefore can be treated as fixed parameters.
\begin{lemma} \label{GD}
Let $NPV^*(P, FiT)$ be the maximized NPV for a given $P, FiT$ under constant demand. Then,
    \begin{equation*}
        NPV^*(P^{max}, FiT^{min}) \geq NPV^*(P, FiT)
        \quad \forall P\leq P^{max},\quad FiT \geq FiT^{min}.
    \end{equation*}
\end{lemma}
\begin{proof}
    Refer to appendix \ref{lemma proof}.
\end{proof}
Therefore, for a given $(P^{max},FiT^{min})$, the second-level model is transformed into a mixed integer linear program (MILP). The regulator's objective function is evaluated over a discrete set of $(P^{max}, FiT^{min})$ policies through an exhaustive grid search. For each policy tuple, the regulator anticipates the DGC's optimal response and selects the policy that maximizes HES. As the DGC's feasible and bounded MILP is ensured an optimal solution, the proposed numerical solution approach guarantees an overall equilibrium between the discretized first-level and the mixed-integer second-level problems.  Although uniqueness is difficult to establish analytically without additional assumptions \cite{kleinert2021survey}, all sets of parameters utilized in the numerical case study yield each a unique equilibrium.

The DGC's MILP is solved using the Gurobi optimizer version 12.0.1 in Python 3.13.9. The computations are performed using a single compute node with one CPU core and 8 GB of allocated RAM. Solving all the second-level MILPs for a given grid of possible $(P^{max},FiT^{min})$ tuples generally requires less than an hour of runtime.

\section{Case study and data} \label{sec:case study}
We apply the game theoretical framework described in Section \ref{sec:methods} to the case of a real microgrid in Deir Qanoun Ennaher, in Lebanon. 

Lebanon's residential electricity market is characterized by a highly intermittent utility, Electricite du Liban (EDL), and a proliferation of diesel-based microgrids \cite{abi2018energy}. 
Since the 2019 economic crisis, the country has observed an organic increase in distributed renewable energy, namely, household solar PV systems \cite{diab2025crisis}. As of 2023, the installed solar PV capacity is estimated at 1000 MW \cite{LCEC2025SolarPV}.

The microgrid considered in our case study is centered around a 400 kW diesel generator. Around 40\% of the households connected to the microgrid are PV-owners. The number of households, their average PV capacity and characteristics of the diesel generator are summarized in Table \ref{tab:DQEN}. The lifetimes of all considered technologies are shown in Table \ref{tab:tech life}.
To obtain the representative household demand profiles, shown in Figure \ref{fig:rep load prof}, we installed loggers in three households and collected hourly load data from August 2024 to April 2025. The raw logs are available on a public database \cite{dbouk2025electricity}. 
As for solar capacity factor data, we relied on the widely cited website \href{renewables.ninja}{renewables.ninja} \cite{pfenninger2016long, staffell2016using}. Based on clustering techniques widely utilized in electricity planning models to ensure computational tractability \cite{ringkjob2020, maranon2019}, we generate three representative daily load profiles, for summer, winter and spring/fall using the k-means clustering algorithm, as shown in Figures \ref{fig:rep load prof} and \ref{fig:cap fact}. The weights of the three representative days are 0.545, 0.29 and 0.165 for summer, winter and spring/fall respectively.
The inputs and their sources are available in the project's Github repository \footnote{link: \url{https://github.com/molleik/microgrid_MP}}, and summarized in Appendix \ref{app:case study}.

\begin{table}[]
    \centering
    \caption{Data on Deir Qanoun Ennaher}
    \begin{tabular}{l c}
    \hline
      Number of PV-owner households &  250\\
      Number of non-PV-owner households   & 400 \\
      Average installed rooftop PV capacity per PV-owner household (kW)  & 4 \\
      Installed diesel generator capacity (kW) & 400\\
      Remaining lifetime of installed diesel generator (y) & 3\\
      \hline
    \end{tabular}
    \label{tab:DQEN}
\end{table}

\begin{table}[]
    \centering
    \caption{Technology lifetimes \cite{ahmad2020distributed, hatton2024global}}
    \begin{tabular}{l c}
    \hline
    Technology & Lifetime\\\hline
      Solar PV & 20 years\\
      Batteries & 8 years\\
      Diesel generator & 5 years\\
      \hline
    \end{tabular}
    \label{tab:tech life}
\end{table}

\begin{figure}
    \centering
    \begin{subfigure}[b]{0.49\textwidth}
        \centering
        \includegraphics[width=\textwidth]{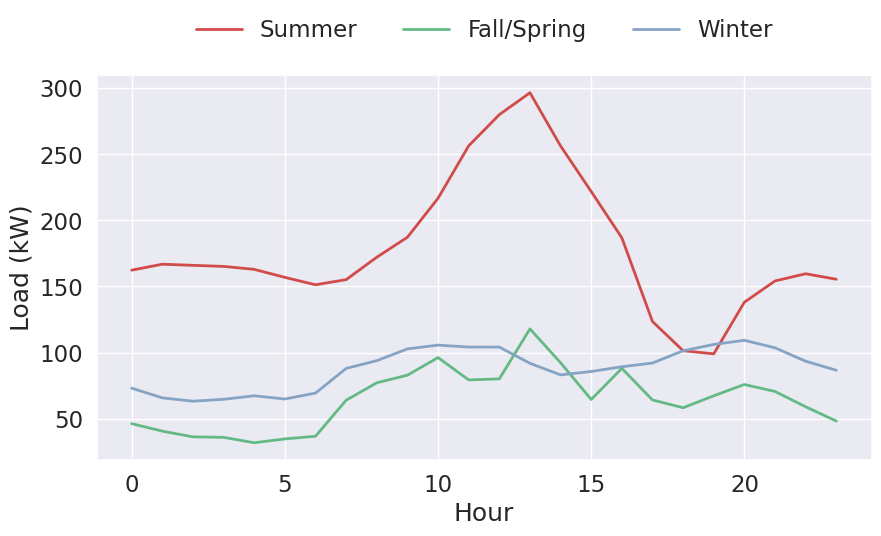}
        \caption{All households owning PV}
        \label{fig:pros}
    \end{subfigure}
    \begin{subfigure}[b]{0.49\textwidth}
        \centering
        \includegraphics[width=\textwidth]{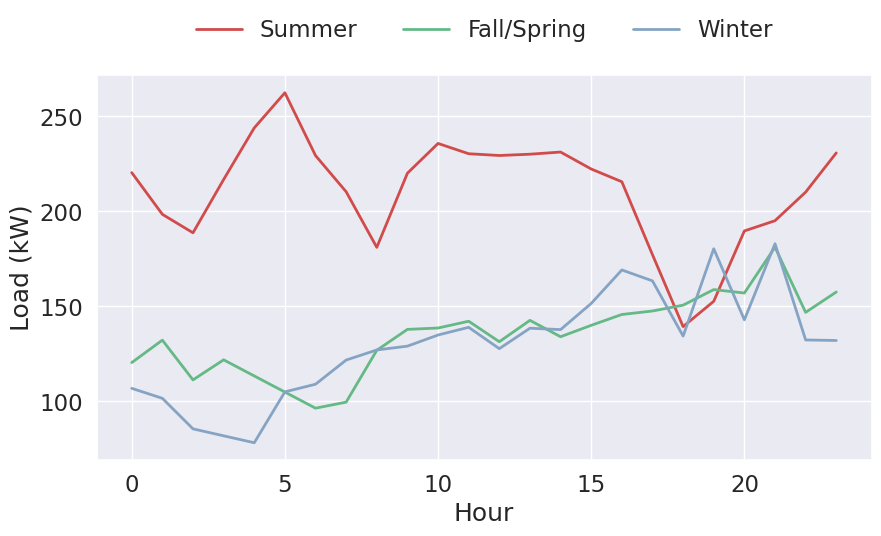}
        \caption{All households not owning PV}
        \label{fig:cons}
    \end{subfigure}
    \caption{Representative microgrid demand profiles}
    \label{fig:rep load prof}
\end{figure}

\begin{figure}
    \centering
    \includegraphics[width=0.5\linewidth]{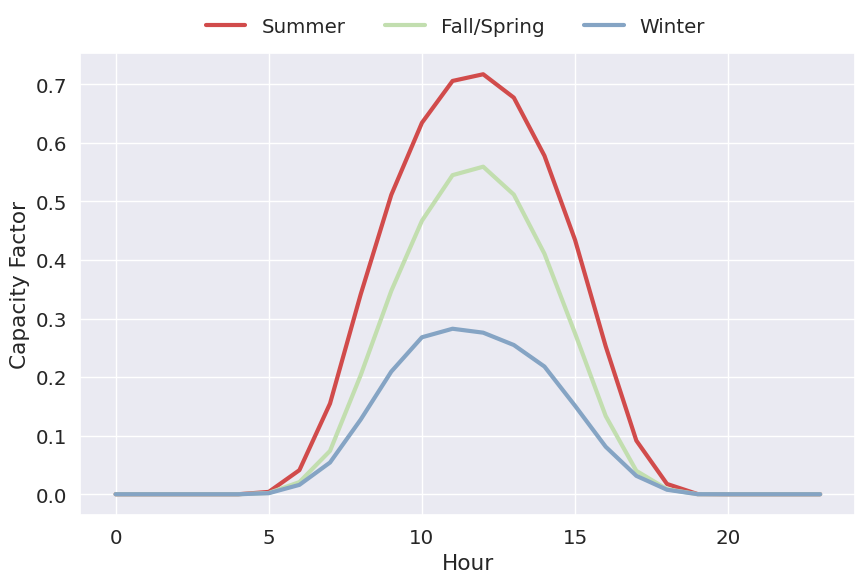}
    \caption{Representative profiles of solar capacity factors}
    \label{fig:cap fact}
\end{figure}

\section{Results and discussion}\label{sec: results}
We first look at the status-quo microgrid, where the DGC only operates a diesel generator without fed-in electricity from household PV-owners' excess, nor installing PV nor batteries. This case constitutes the benchmark against which our proposed model is compared. We compute the current NPV of the DGC's profits ($NPV_0^{DGC}$), which is used as a threshold for the leader's constraint. We also identify the budget needed to support current operations over the planning horizon. This budget, denoted $\Pi_0$, is used as the default budget in the modified model. 

\subsection{Change from status-quo}
Holding the budget at $\Pi_0$, and allowing renewable energy in the microgrid, \textit{i.e.}, fed-in electricity from household PV-owner excess and DGC-owned PV-battery system, the optimal $(P, FiT)$ policy can be found by varying $P$ between 0 and 0.4 USD (the status-quo price) and $FiT$ between 0 USD and $P$. 
Table \ref{tab:benchmark} compares both the status quo and the new model along the following metrics:
\begin{itemize}
    \item The percentage of unmet demand, defined as the ratio of unmet demand to the total demand,
    \item The percentage of wasted excess household PV generation potential, defined as the ratio of wasted excess household PV generation to total household PV generation potential,
    \item The renewable energy penetration, defined as the DGC-owned PV generation and household PV fed-in electricity, as a percentage of the served demand.
\end{itemize}   
The proposed case shows a 18\% improvement in HES. It is driven by a 10\% decrease in electricity price, and the sale of household PV-owner excess generation at 0.12 USD/kWh. The corresponding fed-in electricity results in a decrease of 52.6 percentage points of wasted excess generation potential. Finally, the increased unmet demand in the proposed model is further discussed in Section \ref{sec:ud cons}. These preliminary results already underline the benefits of hybrid microgrids, leading to a more efficient use of the available assets, more affordable electricity, and reduced emissions, despite the existing market power.

\begin{table}[]
    \centering
    \caption{Comparison between status-quo and proposed model, where M refers to millions and k to thousands.}
    \begin{tabular}{ccc} \hline
     & Status quo & Proposed model \\\hline
    $NPV^{DGC}$ (USD) & 2.25 M & 2.46 M \\
    Budget (USD)& 420 k & 420 k \\
    HES (USD)& 6.99 M & 8.26 M \\
    Unmet demand (\%)& 1.9 & 12.2 \\
    Wasted excess PV generation potential (\%)& 58.8 & 6.2 \\
    Price (USD/kWh)& 0.4 & 0.36 \\
    Feed-in-tariff (USD/kWh)& - & 0.12 \\
    Renewable energy penetration (\%) & 0 & 60.1\\\hline
    \end{tabular}
    \label{tab:benchmark}
\end{table}

Figure \ref{fig:FiTs v Prices 400} highlights the feasible policy tuples, for which $NPV^{DGC}\geq NPV_0^{DGC}$. The drawn curve maps, for every $P$, the maximum $FiT$ that keeps the DGC satisfied with profits equal to the ones in the current case. As the $FiT$ on the boundary increases, it becomes high enough for the DGC to drastically reduce fed-in electricity, and rely on owned PV generation instead.
This case is represented by the sharp increase in the maximum feasible $FiT$ after $P = 0.37$ USD/kWh. The hatched area under this portion of the curve shows the policy tuples where the fed-in household PV excess is limited and makes up less than 5\% of the total supply to the microgrid. The remaining highlighted region represents the set of $(P, FiT)$ tuples for which significantly feeding household PV excess into the system is economically viable for the DGC. From an economic standpoint, allowing the DGC to feed-in electricity and build its own PV system provides it with more options and potentially reduces its generation costs. Consequently, the DGC will be able to maintain its current profits while paying a $FiT$ to household PV-owners or reducing the sale price of electricity $P$. The boundary represents the iso-profit curve for the DGC such that each point is a reflection of its willingness to pay for feed-in when the electricity sale price is set at a given $P$. 

For every feasible $(P, FiT)$ tuple, the HES is shown in Figure \ref{fig:HES}. Considering Figures \ref{fig:FiTs v Prices 400} and \ref{fig:HES}, the optimal and near-optimal tuples include significant fed-in electricity from household PV-owner excess. The equilibrium found at $P=0.36$ and $FiT=0.12$ is a direct result of the tradeoff between $FiT$, as effectively decided by the regulator, and the fed-in electricity $Fi$ chosen by the profit-maximizing DGC. As neither metric is driven to an extreme, it can be suggested that the value placed by the DGC on the household PV-owner excess gives the regulator some influence over the operational and investment decisions in the microgrid.

\begin{figure}[]
\centering
    \begin{subfigure}[b]{0.4\textwidth}
        \centering
        \includegraphics[width=\textwidth]{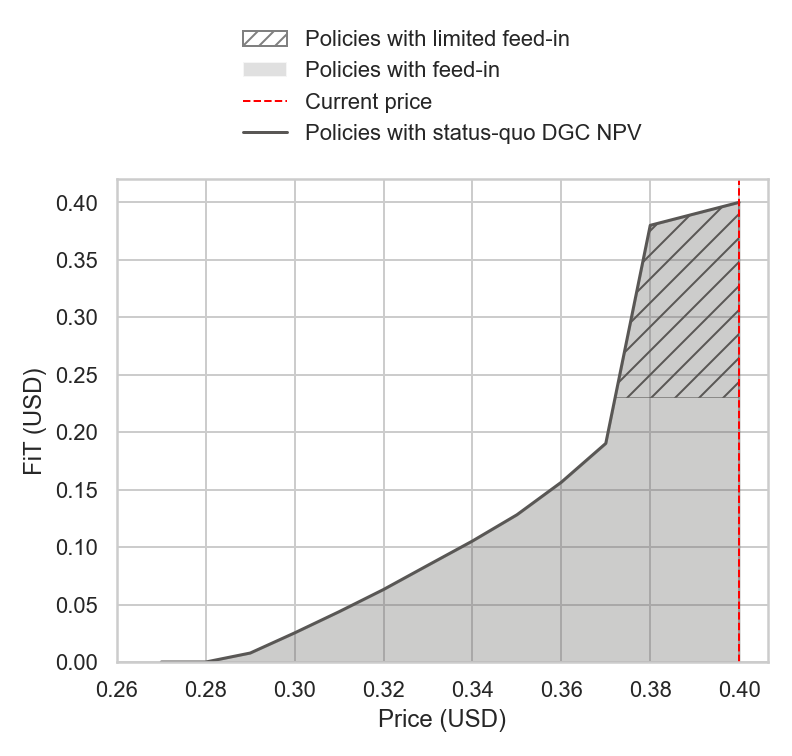}
        \caption{Feasibility region}
        \label{fig:FiTs v Prices 400}
    \end{subfigure}
    \begin{subfigure}[b]{0.4\textwidth}
        \centering
        \includegraphics[width=\textwidth]{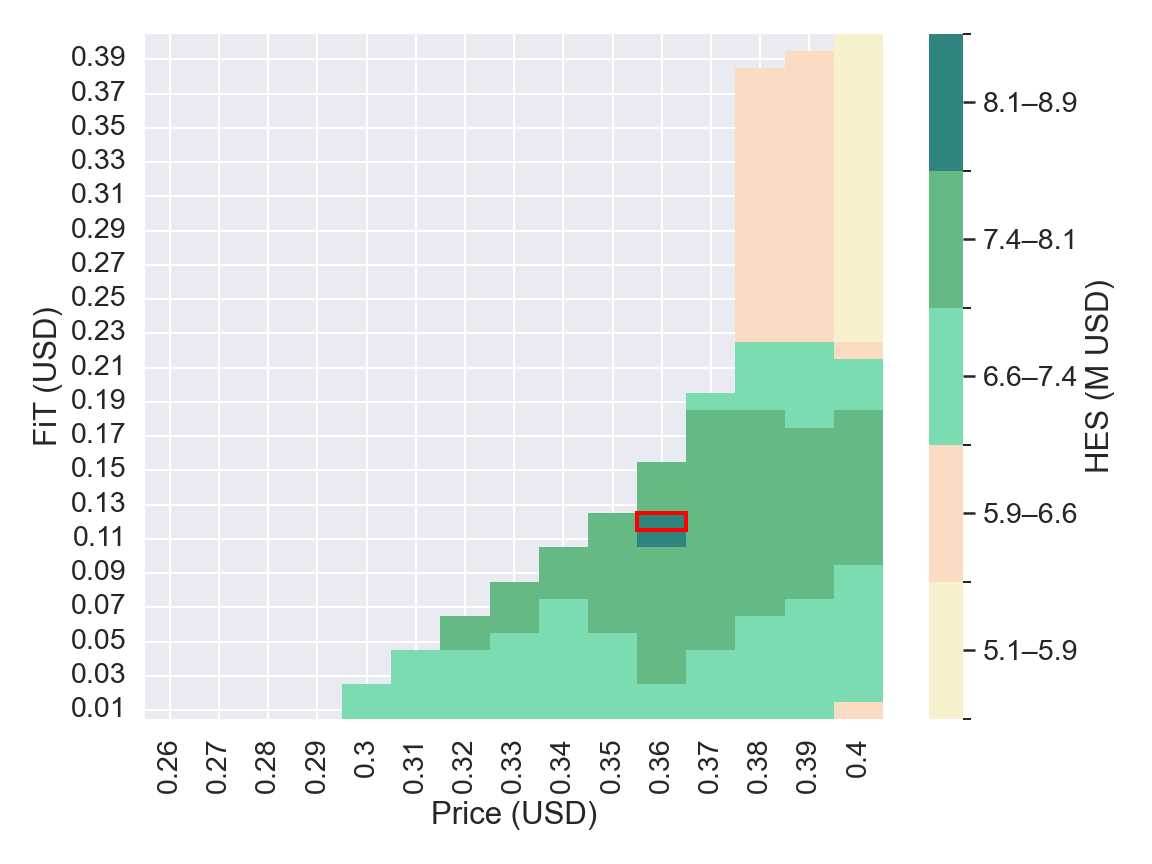}
        \caption{HES (maximum inside red box)}
        \label{fig:HES}
    \end{subfigure}
    \caption{Feasible region and HES for every $(P, FiT)$ tuple in the proposed model}
    \label{fig:proposed model HES}
\end{figure}

The HES resulting from the proposed model is a reflection of the change in the capacity and generation portfolio in the microgrid. Figure \ref{fig:inst_cap_11_35} shows the yearly in-place capacities. During the first three years of operations, the DGC uses the existing diesel capacity, holding off on any additional investment in a PV-battery system until the end of the diesel generator's lifetime. Then, the lower LCOE of PV-battery is more attractive than the relatively expensive diesel generator. With the installation of batteries, the DGC also feeds-in more PV-owner excess (Figure \ref{fig:year_gen_11_35}). The addition of a PV-battery system leads to unmet demand during some hours in winter. The available PV capacity and battery charge cannot meet the demand, and the expansion of any capacity for the purpose of serving these few hours is prohibitively expensive. 

\begin{figure}
\centering
    \begin{subfigure}[b]{0.49\textwidth}
        \centering
        \includegraphics[width=\textwidth]{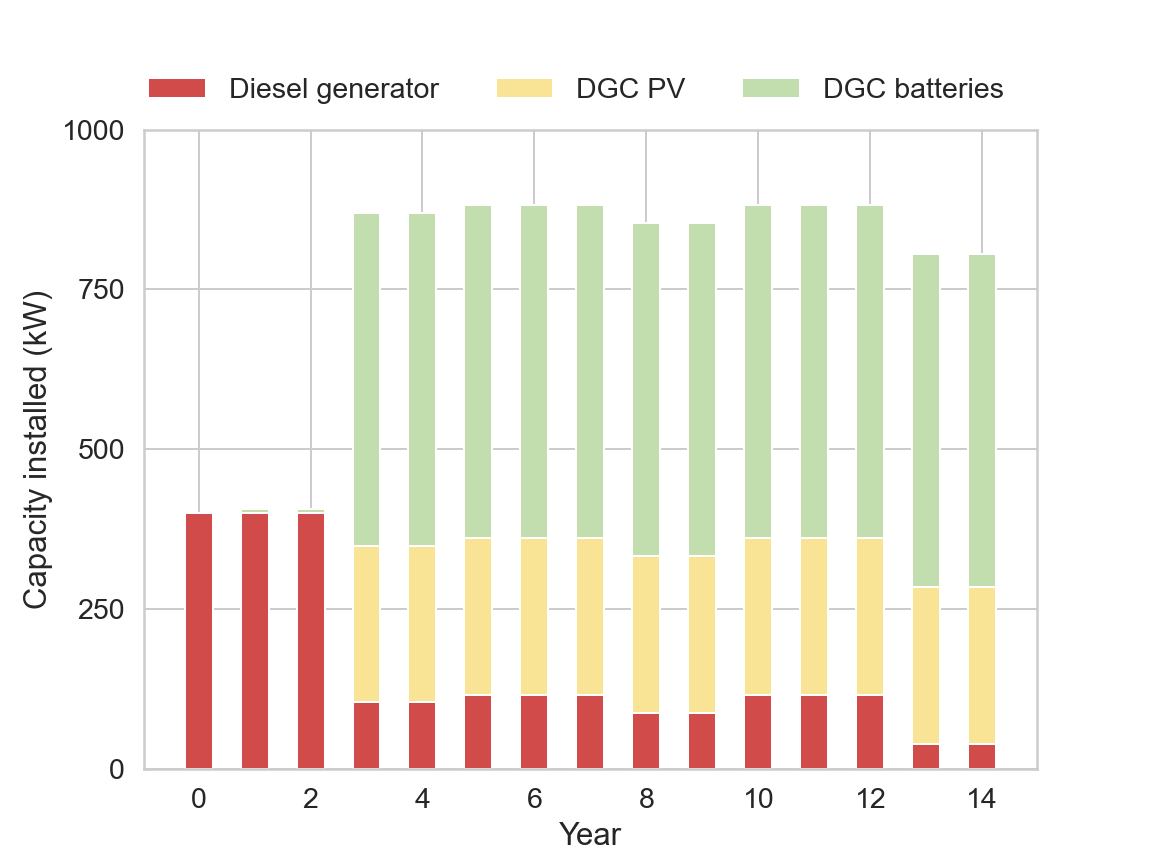}
        \caption{Installed capacities}
        \label{fig:inst_cap_11_35}
    \end{subfigure}
    \begin{subfigure}[b]{0.49\textwidth}
        \centering
        \includegraphics[width=\textwidth]{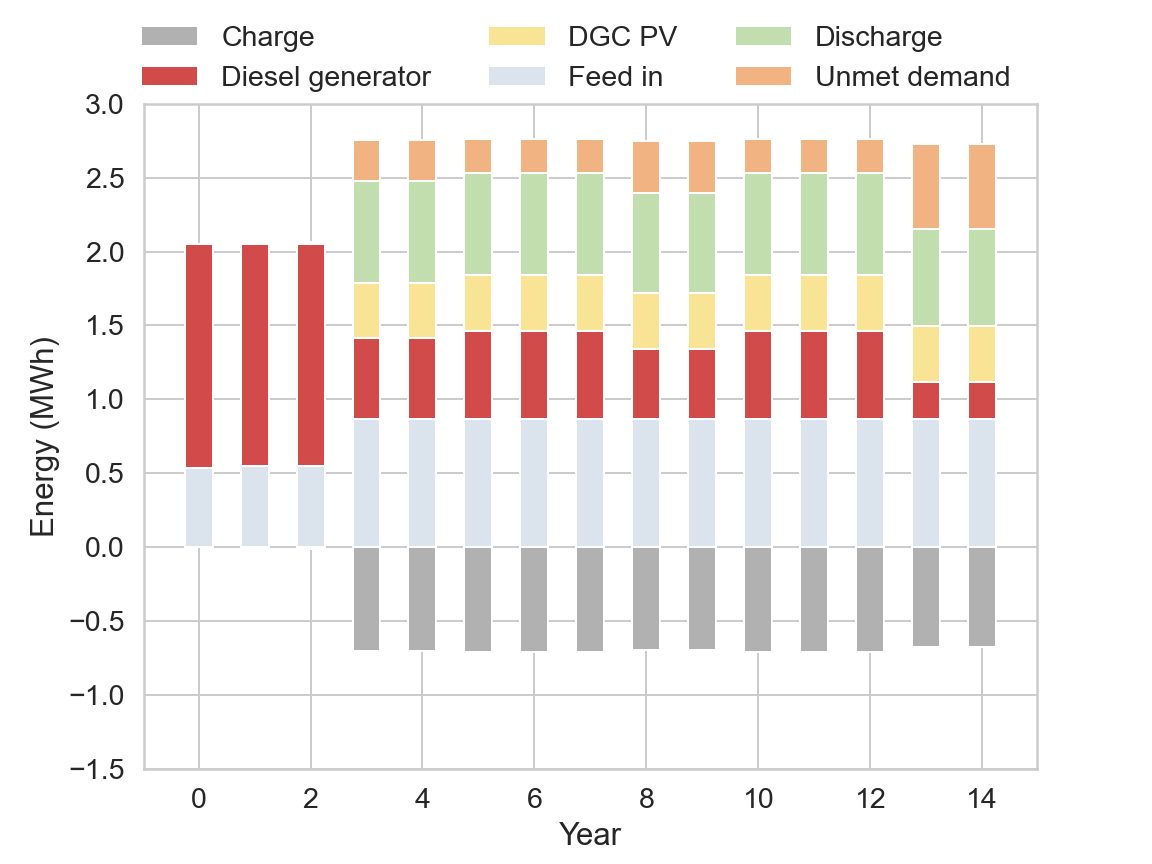}
        \caption{Yearly generation profile}
        \label{fig:year_gen_11_35}
    \end{subfigure}
    \caption{Capacity and generation portfolio in the proposed model}
    \label{fig:proposed model tech}
\end{figure}

\subsection{Sensitivity analysis on the budget constraint} \label{sec:budget sens}
Financing generation capacity expansion is a key concern when determining the optimal size of a microgrid, particularly in lower-income countries. As investment budgets often constitute a barrier to achieving grid reliability, the literature commonly acknowledges this constraint \cite{afful2017power}. As long as the DGC's decisions are bound by the available investment budget, the regulatory entity's consideration for this constraint is essential when evaluating the HES-maximizing ($P$, $FiT$) tuple. 

\begin{figure}
    \centering
    \includegraphics[width=0.5\textwidth]{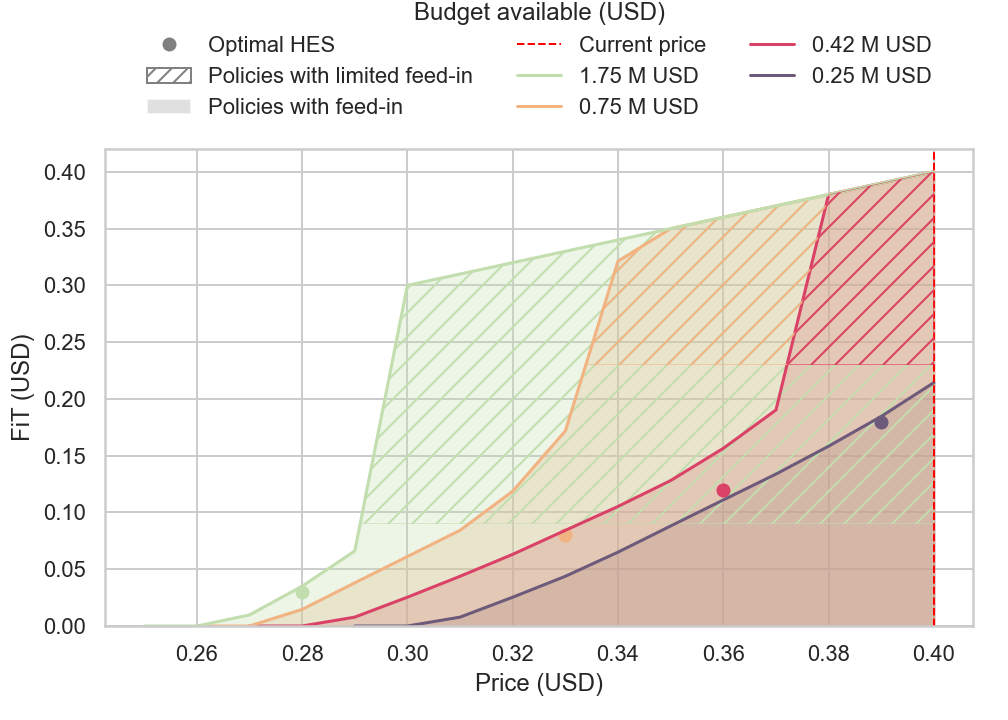}
    \caption{Feasible regions of regulatory entity policies}
    \label{fig:fits v prices}
\end{figure}

Similar to Figure \ref{fig:FiTs v Prices 400}, Figure \ref{fig:fits v prices} shows the feasible $(P, FiT)$ regions under different budgets. For the budget of 1.75 M USD, we find that the budget constraint becomes non-binding for all $(P, FiT)$ tuples, and the DGC significantly reduces its reliance on household PV excess as soon as the $FiT$ exceeds the LCOE of investing in additional PV capacity, corresponding to 0.08 USD/kWh. For lower budgets, the tipping $FiT$ between feeding-in electricity from households and almost fully relying on DGC owned PV system increases to 0.23 for budgets of 0.75 M and 0.42 M, and over 0.40 USD/kWh for a budget of 0.25 M. As it is the case under $\Pi_0$, for each of the considered budgets, the optimal $(P, FiT)$ tuple still allows for significant fed-in electricity from PV-owner excess as this tuple is always located in the non-hatched region.

Figure \ref{fig:household surplus} shows the HES for all feasible $(P, FiT)$ tuples under different budgets. In all cases, the optimal microgrid configuration promotes the efficient utilization of household assets. Moreover, these equilibria imply that, under binding budgets, the household PV-owner excess is a more valuable resource because it does not entail additional investments, and increases electricity availability in the microgrid. With higher budgets, the DGC's set of feasible investment decisions expands, reducing reliance on fed-in electricity. Naturally, from a regulator perspective, when the DGC budget is tight, the $FiT$ component of its policy gains higher importance and allows setting higher values as reflected in the non-hatched areas of Figure \ref{fig:fits v prices}. With higher DGC budgets, the $FiT$ component is less valuable; however, the regulator is able to further reduce the electricity price component $P$ in the microgrid.

\begin{figure}
    \centering
    \begin{subfigure}[b]{0.41\textwidth}
        \centering
        \includegraphics[width=\textwidth]{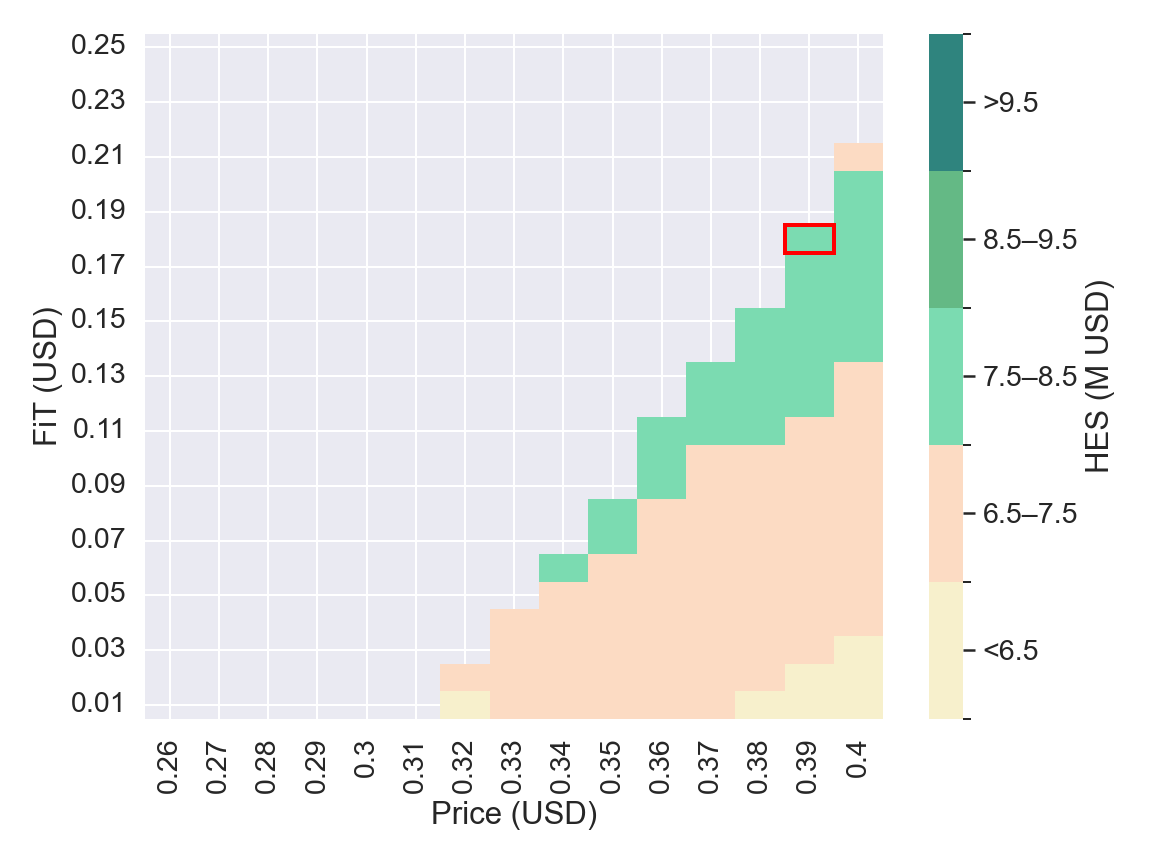}
        \caption{Budget = \$250k }
        \label{fig:Fi (HS)}
    \end{subfigure}
    \begin{subfigure}[b]{0.41\textwidth}
        \centering
        \includegraphics[width=\textwidth]{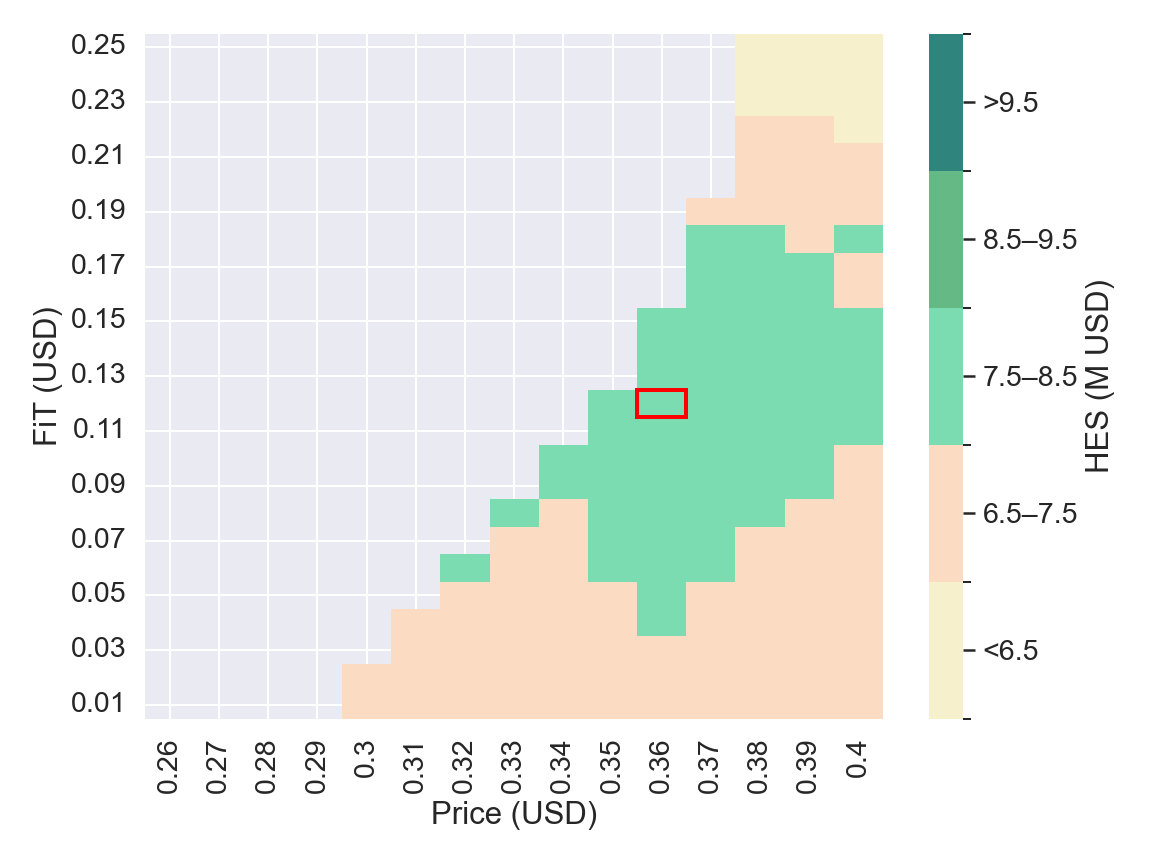}
        \caption{Budget = \$420k}
        \label{fig:Fi + Bat (HS)}
    \end{subfigure}
    \begin{subfigure}[b]{0.41\textwidth}
        \centering
        \includegraphics[width=\textwidth]{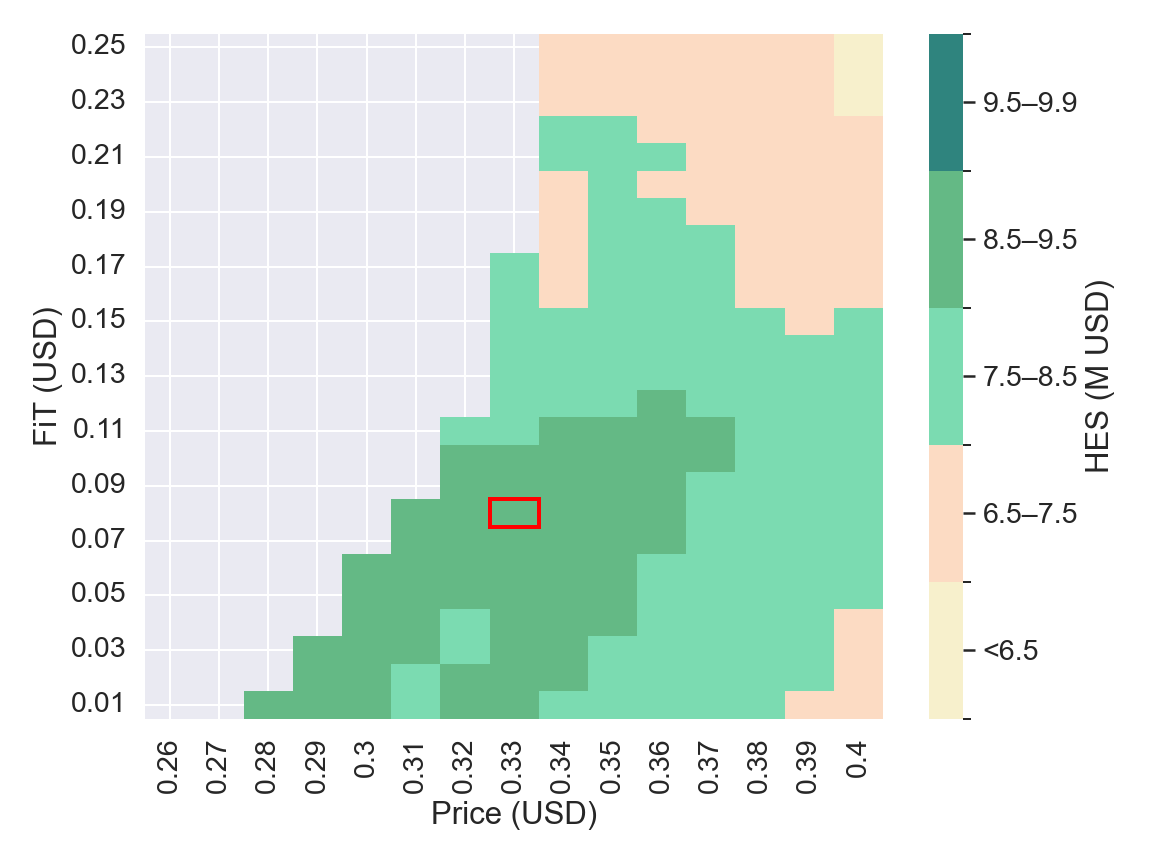}
        \caption{Budget = \$750k }
        \label{fig:Fi + Bat + PV (HS)}
    \end{subfigure}
    \begin{subfigure}[b]{0.41\textwidth}
        \centering
        \includegraphics[width=\textwidth]{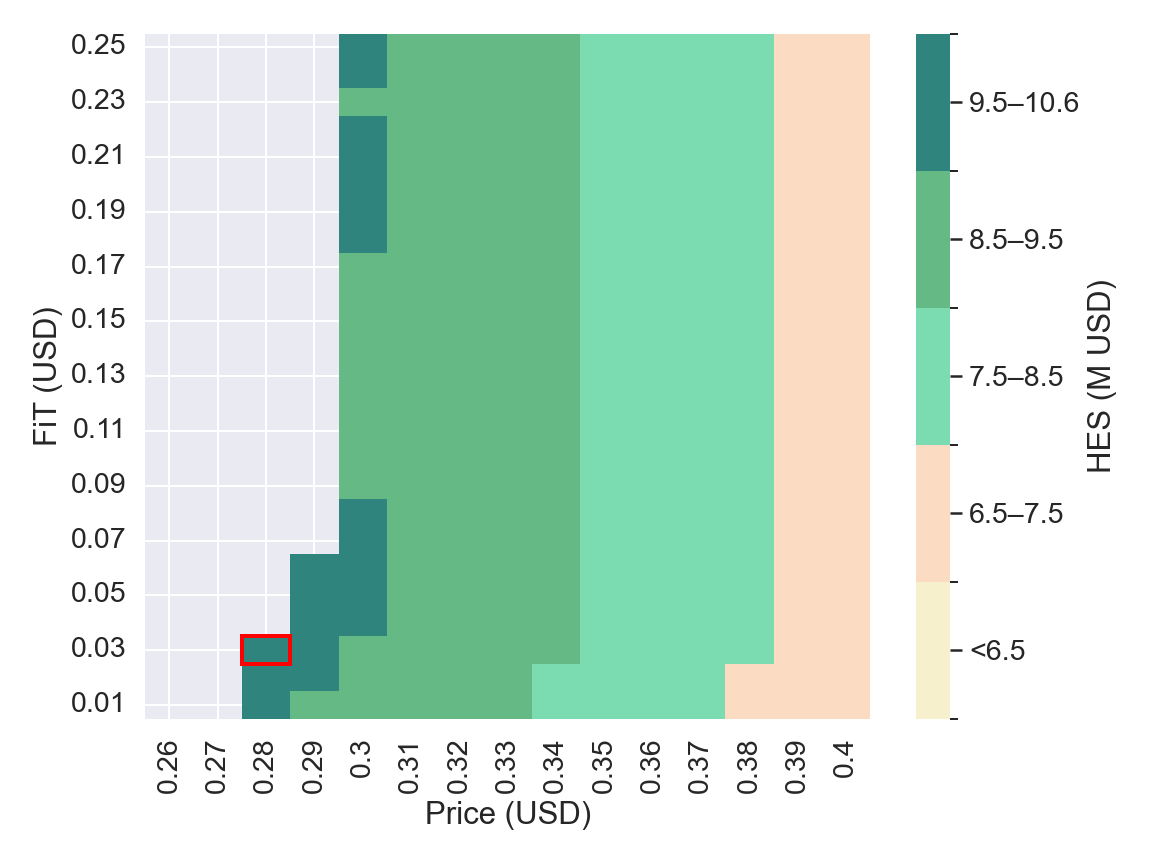}
        \caption{Budget = \$1.75M}
        \label{fig:Fi + Bat + PV (HS)}
    \end{subfigure}
    \caption{HES of feasible policies, with red boxes indicating the maximum}
    \label{fig:household surplus}
\end{figure}

Higher budgets generally yield higher HES (Figure \ref{fig:household surplus v budget}). Under the proposed model, the status-quo HES can be recovered for a budget lower than $\Pi_0$ by 58.3\%. Maintaining the benchmark budget $\Pi_0$, the HES would increase by 18\% with the proposed model. The HES increase is not limited to the optimal policies; the darker regions in Figure \ref{fig:household surplus} show that other $(P, FiT)$ tuples generally benefit from budget increases as the addition of capital-intensive PV-battery capacity results in more served demand or more fed-in electricity from household PV-owner excess.

\begin{figure}
    \centering
    \includegraphics[width=0.5\linewidth]{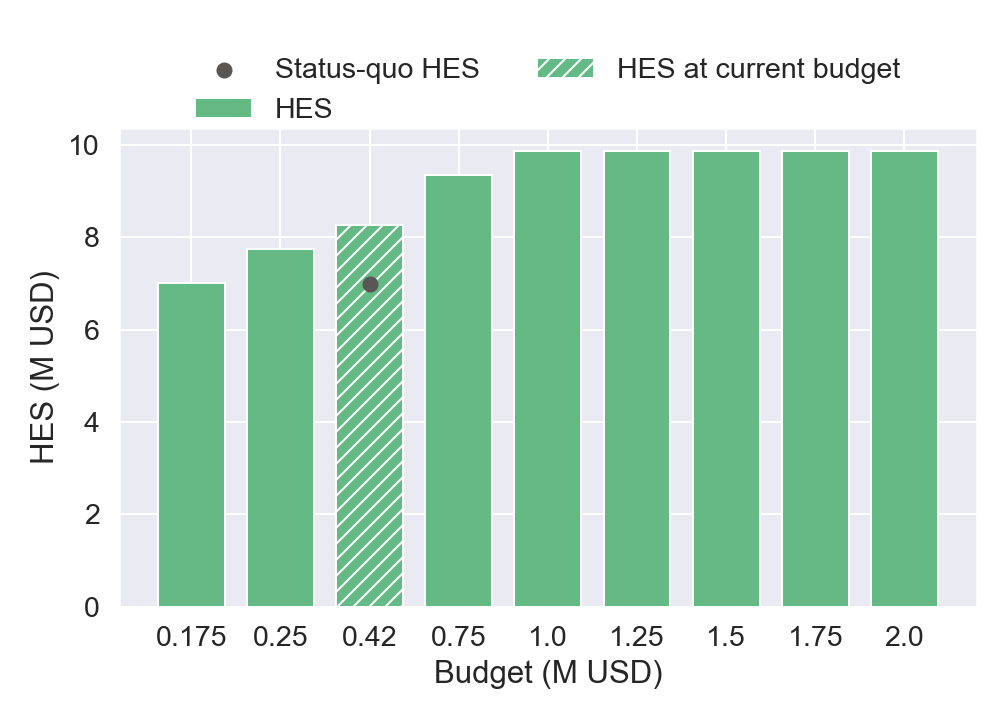}
    \caption{Optimal HES at different budgets}
    \label{fig:household surplus v budget}
\end{figure}

For every optimal HES, at each budget, the four values driving the regulator's objective function: $P$, $FiT$, served demand, and household PV-owner excess utilization can be studied more closely.
Figure \ref{fig:P-FiT v budget} shows the optimal prices and feed-in-tariffs for every considered case. As the budget increases, the DGC can commit higher investments to replace diesel-based capacity characterized by a high LCOE with more economical PV-battery systems, reducing the overall costs and allowing for a price decrease, in turn increasing the HES. The $FiT$ decreases as well to remain competitive with the increasing PV-battery capacity allowed by higher budgets.

\begin{figure}
    \centering
    \includegraphics[width=0.5\linewidth]{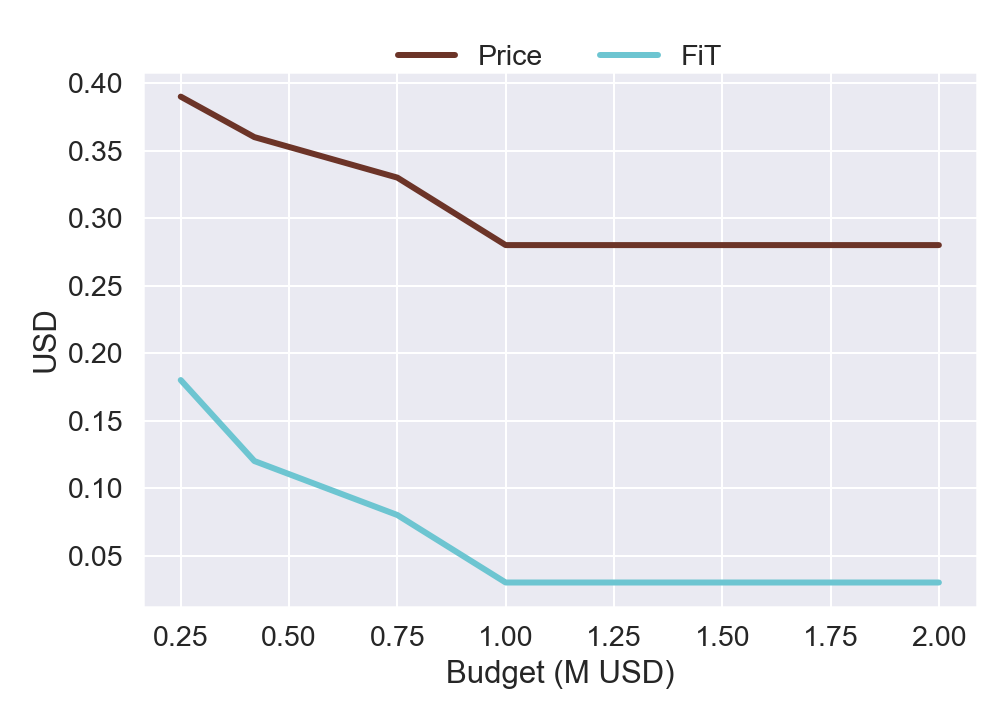}
    \caption{Optimal price and feed-in-tariff at different budgets}
    \label{fig:P-FiT v budget}
\end{figure}

These findings are further reflected in Figure \ref{fig:baseline cap}. With tight budgets, additional diesel capacity is needed. Batteries are first introduced to increase the utilization of the household PV-owner excess available for feed-in. When the budget reaches 0.75 M USD or more, sizable PV investments are undertaken. At budgets of 1 M USD and above, diesel generator capacities are entirely phased out following the retirement of existing units. Figure \ref{fig:baseline ene} highlights the considerable share of total generation that the fed-in electricity constitutes. At a budget of 0.75 M USD, the large installed capacity of battery storage results in maximizing the contribution of fed-in electricity to the microgrid demand to 39\% and in minimizing the wasted household PV excess generation potential to 5\%. For higher budgets, the increased PV capacity then reduces the utilization of fed-in electricity, resulting in wasted household PV excess generation potential to reach 10\%. For the tightest budget, the unmet demand share reaches 15\%. As the budget is further relaxed, this share drops to a stable 4\%.

\begin{figure}
\centering
    \begin{subfigure}[b]{0.49\textwidth}
        \centering
        \includegraphics[width=\textwidth]{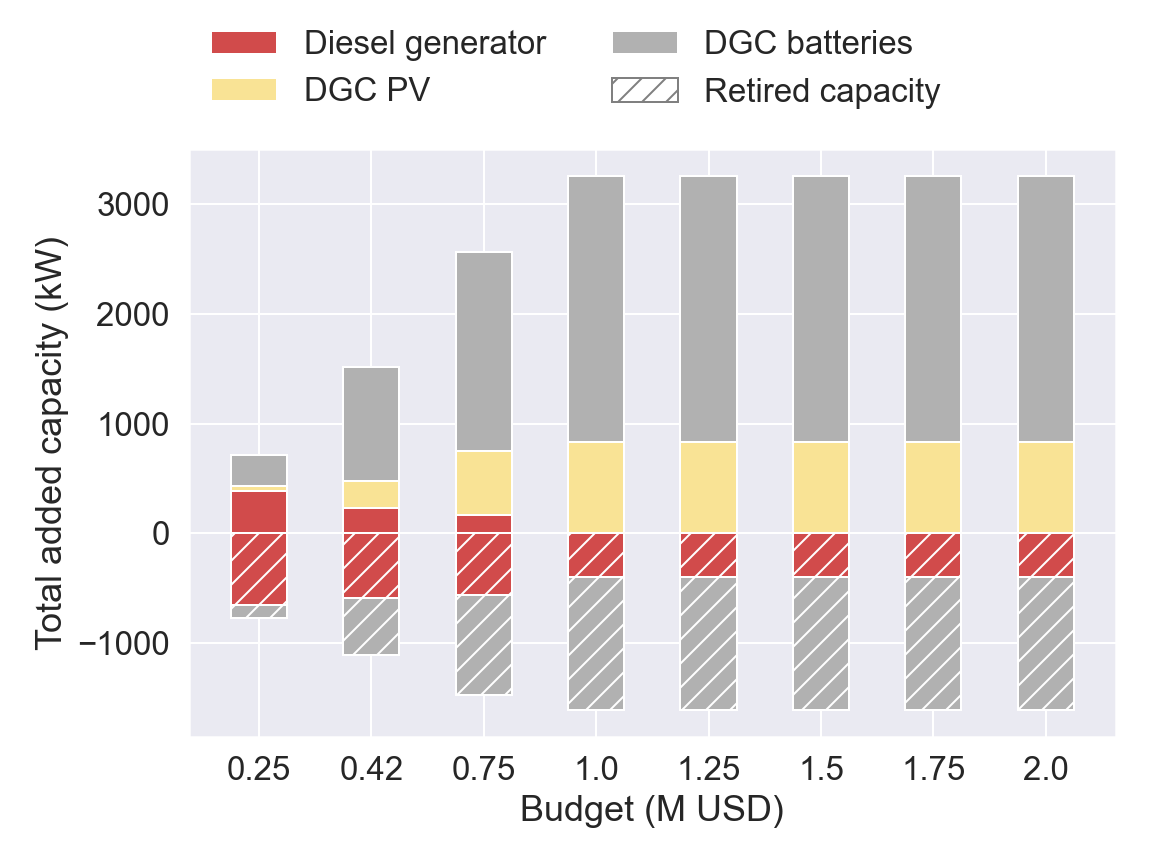}
        \caption{Total added/retired capacity}
        \label{fig:baseline cap}
    \end{subfigure}
    \begin{subfigure}[b]{0.49\textwidth}
        \centering
        \includegraphics[width=\textwidth]{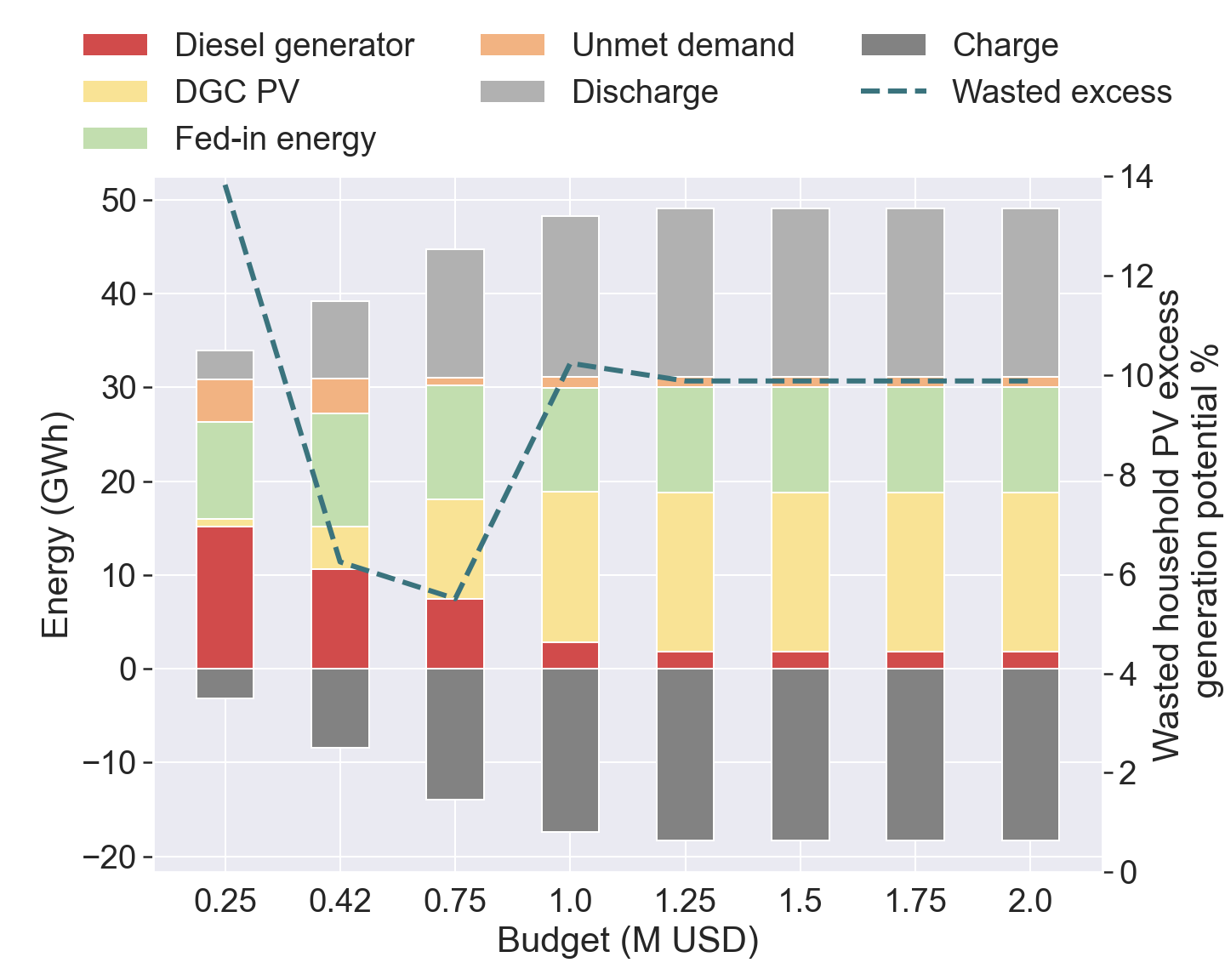}
        \caption{Total energy generation}
        \label{fig:baseline ene}
    \end{subfigure}
    \caption{Total added capacities and total energy for different budgets}
    \label{fig:cap and ene}
\end{figure}

\subsection{Sensitivity on number of household PV-owners}\label{sec:pros sens}
PV-owner participation is a key factor for HES improvement. It is therefore imperative to understand the effects of varying the share of household PV-owners in the microgrid on the HES. For every household PV-owner penetration case, corresponding benchmark NPV and budget are computed. The change in HES between each microgrid's status-quo and the corresponding proposed model is illustrated in Figure \ref{fig:change_pros}. At 0\% household PV-owners, introducing a PV capacity upon retirement of the diesel generator lowers operating costs and increases the DGC’s NPV, even though unmet demand rises (Figure \ref{fig:cap and ene @ pros}). Coupled with the absence of fed-in electricity, this results in a decrease in the HES from the status quo. With higher shares of PV-owners, this undesired outcome is mitigated. The HES increases with the increasing penetration of PV-owners until being maximized at a 90\% penetration rate. Understandably, the percentage of wasted household PV generation potential increases with the share of PV-owners, due to both the decrease in net demand on the microgrid, and the increase in the total generation capacity of households (Figure \ref{fig:baseline ene @ pros}). It is also worth noting that, as the PV-owner share reaches 75\% or higher, the DGC's operations are almost completely focused on battery storage.

\begin{figure}
\centering
        \includegraphics[width=0.5\textwidth]{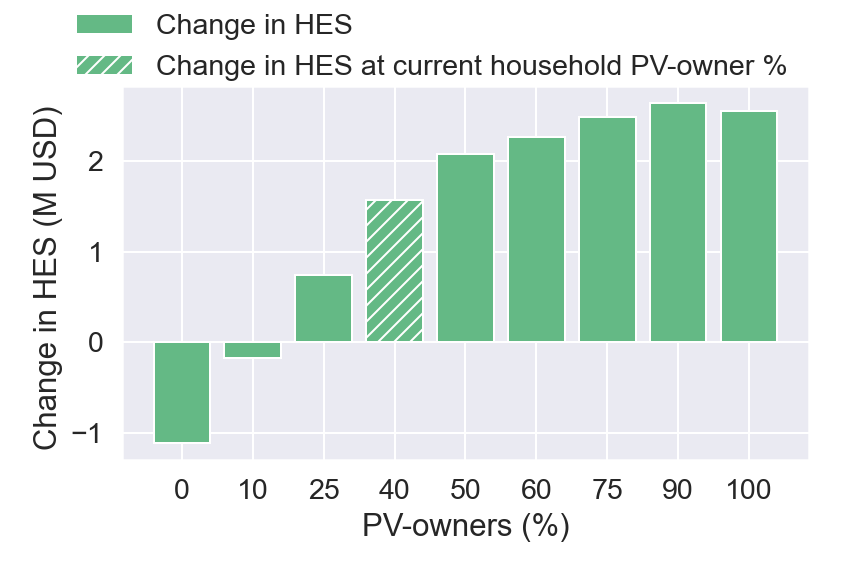}
        \caption{Difference in HES between the proposed model and status-quo}
        \label{fig:change_pros}
\end{figure}

\begin{figure}
\centering
    \begin{subfigure}[b]{0.49\textwidth}
        \centering
        \includegraphics[width=\textwidth]{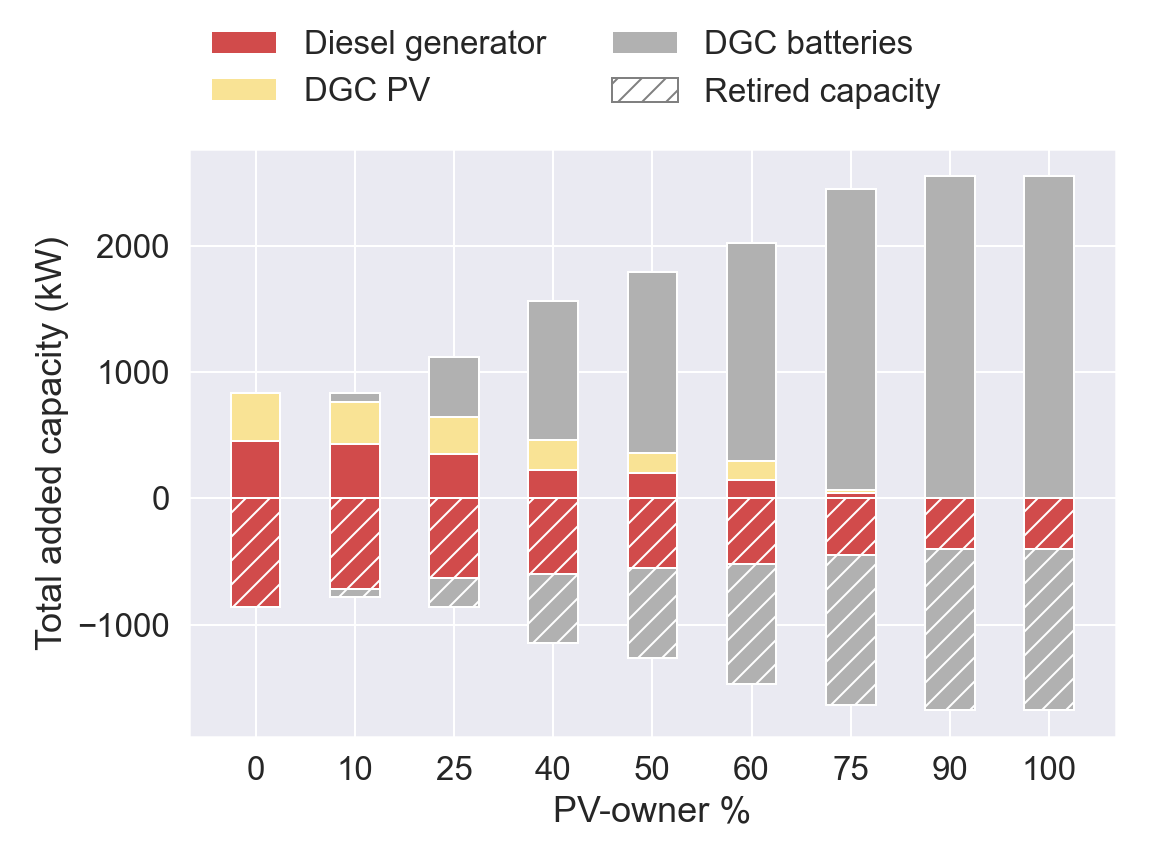}
        \caption{Total added/retired capacity}
        \label{fig:baseline cap @ pros}
    \end{subfigure}
    \begin{subfigure}[b]{0.49\textwidth}
        \centering
        \includegraphics[width=\textwidth]{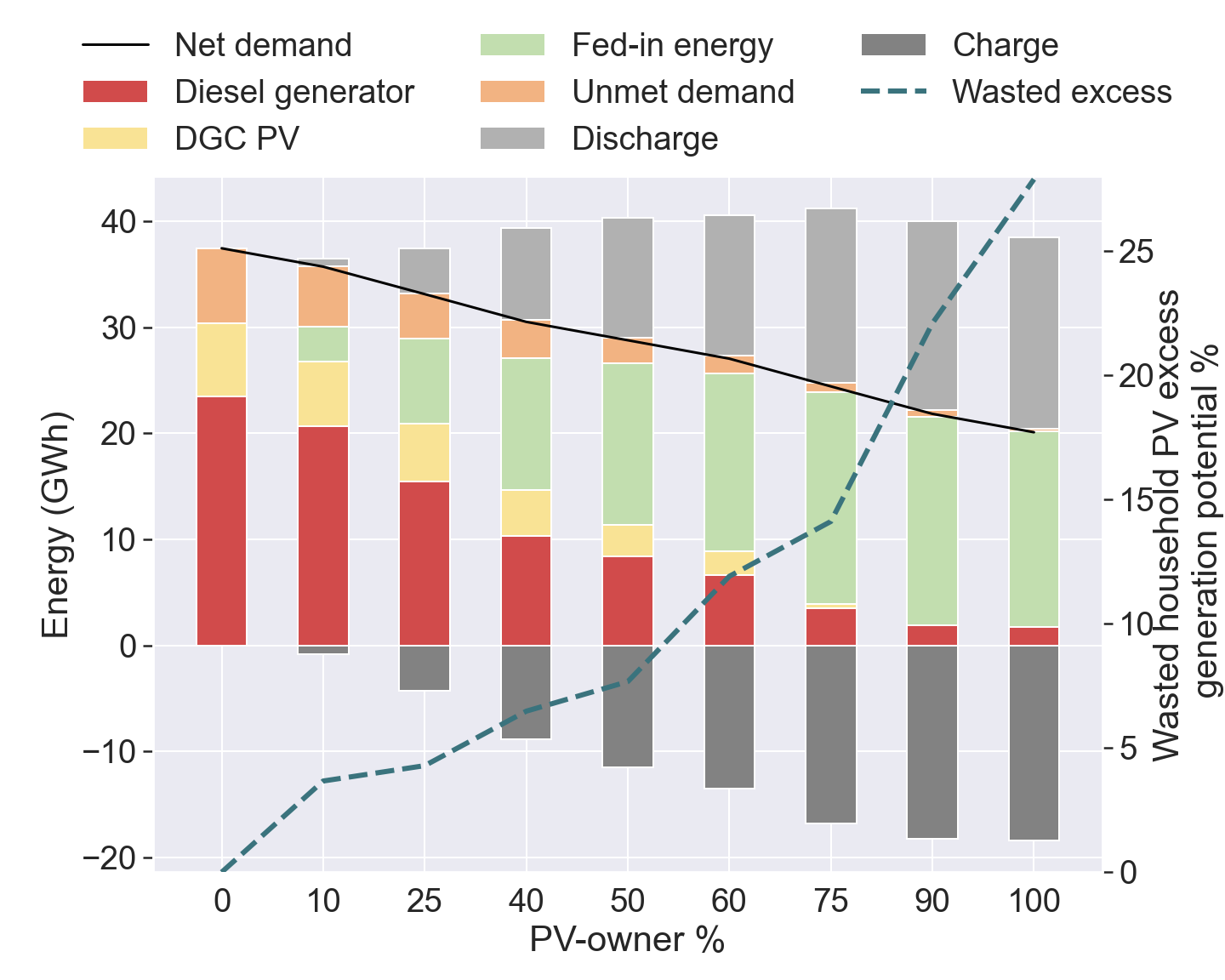}
        \caption{Total energy generation}
        \label{fig:baseline ene @ pros}
    \end{subfigure}
    \caption{Added capacities and energy for different shares of PV-owners}
    \label{fig:cap and ene @ pros}
\end{figure}

\subsection{Constraining the unmet demand} \label{sec:ud cons}
As shown in previous sections, the proposed microgrid model results in limited unmet demand except when the budget is very tight or when the share of PV-owners in the microgrid is low. To further shed light on this important measure, Figure \ref{fig:UD} compares the unmet demand of the proposed model to that of the status quo for different budgets (Figure \ref{fig:UD budget}) and shares of PV-owners in the microgrid (Figure \ref{fig:UD pros}).

It is noticeable that the proposed microgrid model results in higher unmet demand compared to the status quo for almost all the cases considered. 
Regulating the $(P^{max}, FiT^{min})$ tuple alone, while permitting a profit-maximizing DGC to access fed-in electricity from household PV-owners and invest in PV-battery systems, generally results in a low unmet demand but is insufficient to keep it at levels achievable in diesel-only microgrids. This observation reflects the regulator’s limited oversight over supply quality in the microgrid, which cannot be addressed through price caps alone. 

\begin{figure}
\centering
    \begin{subfigure}[b]{0.49\textwidth}
        \centering
        \includegraphics[width=\linewidth]{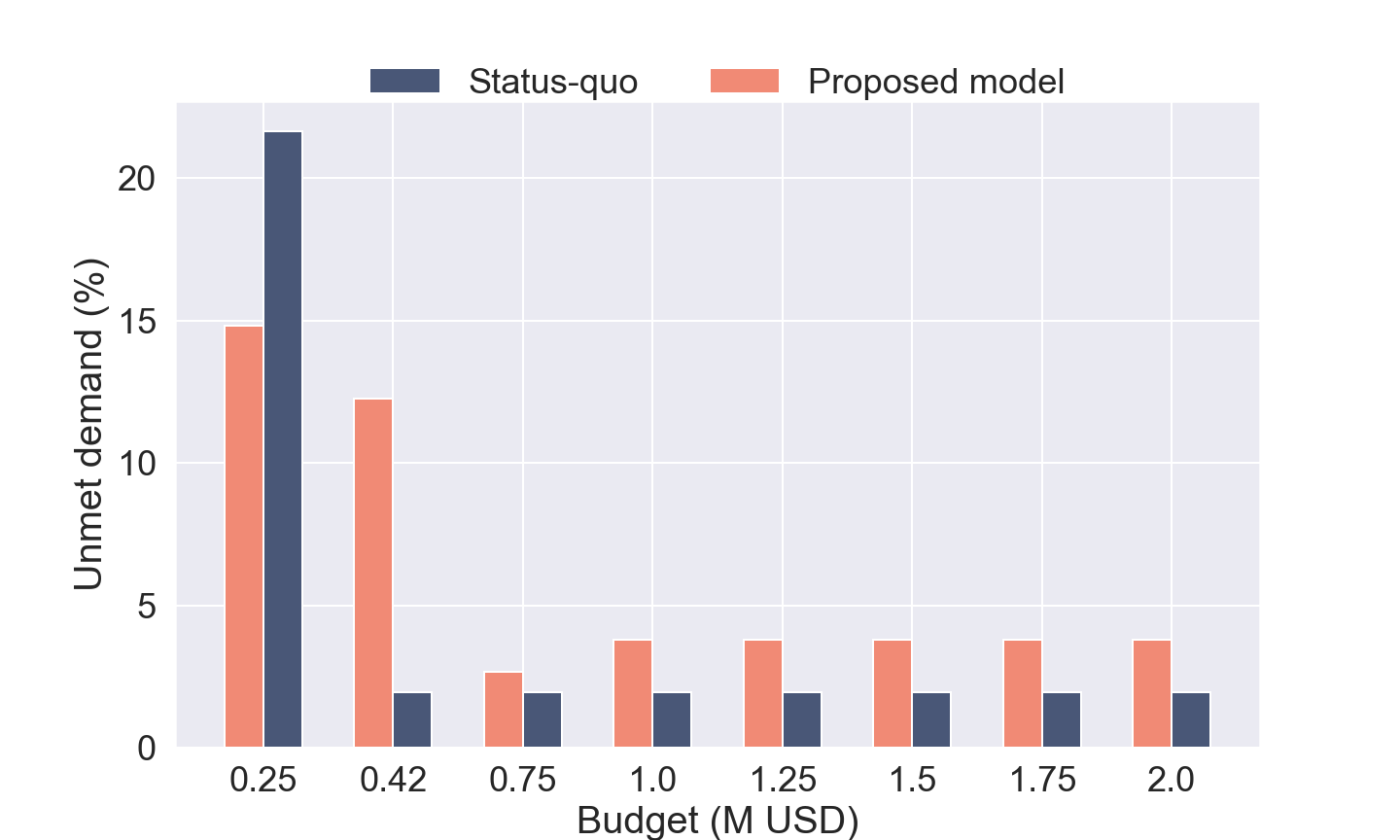}
    \caption{Unmet demand under different budgets}
    \label{fig:UD budget}
    \end{subfigure}
    \begin{subfigure}[b]{0.49\textwidth}
        \centering
        \includegraphics[width=\linewidth]{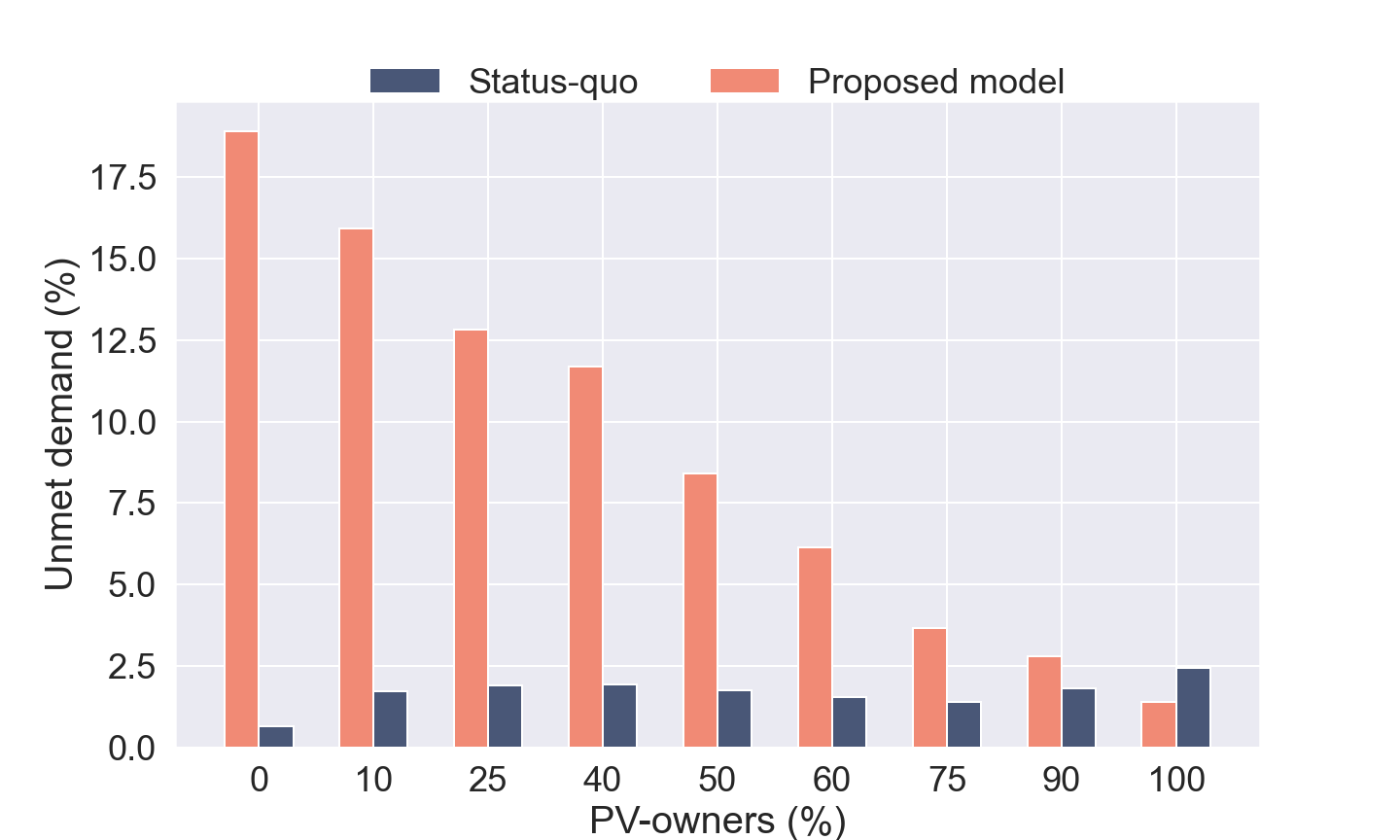}
    \caption{Unmet demand under different PV-owner shares}
    \label{fig:UD pros}
    \end{subfigure}
    \caption{Unmet demand for proposed model and status quo}
    \label{fig:UD}
\end{figure}

To quantify the effects of the limited regulatory oversight over the quality of supply, and therefore the HES, we introduce a theoretical regulator with the extended power to impose, in addition to $(P^{max}, FiT^{min})$, the status-quo unmet demand levels. Let  $\chi_0$ represent the unmet demand in the status-quo microgrid. The following constraint guarantees that the unmet demand in the proposed model does not exceed $\chi_0$ (\ref{eq:ud limit}):
\begin{equation} \label{eq:ud limit}
    \sum_y \sum_d \omega_d \sum_h U_{y,d,h} \leq \chi_0
\end{equation}

Figure \ref{fig:HES cons comp} compares the HES for three cases: (i) the status-quo, (ii) the proposed model where the regulator only controls the tuple $(P^{max}, FiT^{min})$, and (iii) the proposed model with extended regulation imposing Equation \ref{eq:ud limit} under different budgets and shares of PV-owners in the microgrid. The following observations can be made. As revealed in previous results, the proposed model strongly outperforms the status-quo except for cases where the PV-owner share in the microgrid is 10\% or lower. The extended regulation is always performing best. The difference between limited and extended regulations is less than 5\% for budgets at or above 0.75 M USD and less than 7.6\% for shares of PV-owners at or above 40\%. This highlights that, despite the limitations of only relying on price caps in achieving lower unmet demand levels, the loss in HES in the proposed model compared to a regulator with the power to impose a certain quality of supply is small.

\begin{figure}
\centering
    \begin{subfigure}[b]{0.495\textwidth}
        \centering
        \includegraphics[width=\linewidth]{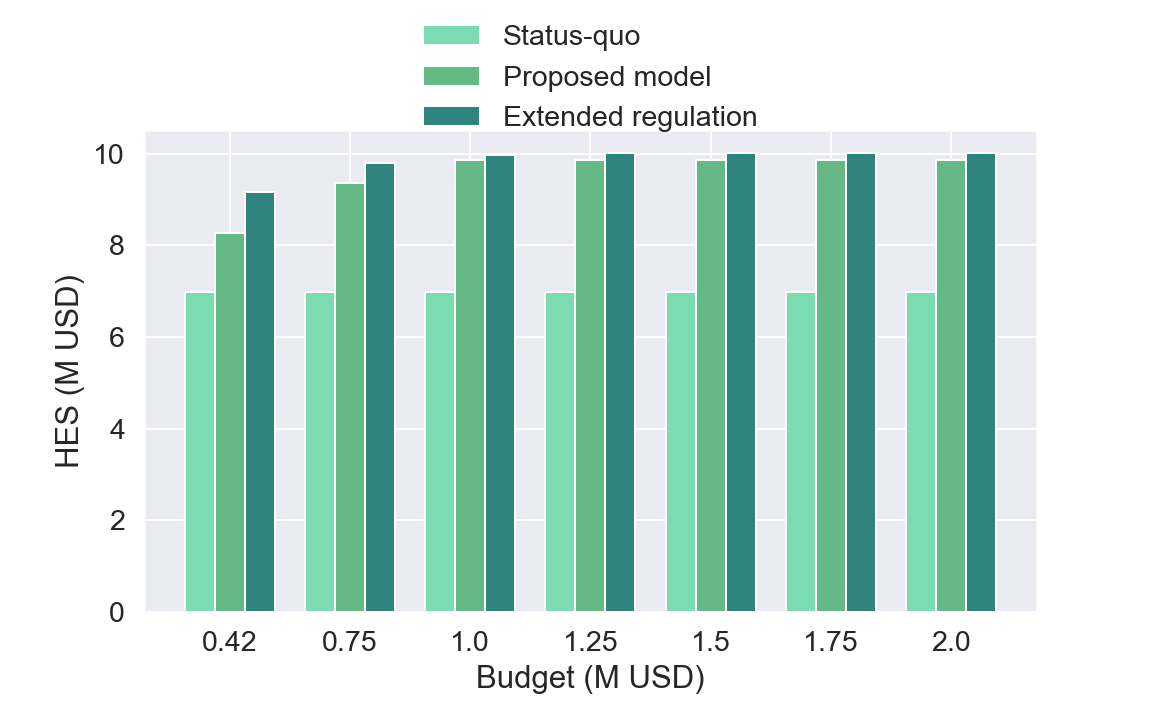}
    \caption{HES under different budgets}
    \label{fig:HES cons comp budget}
    \end{subfigure}
    \begin{subfigure}[b]{0.495\textwidth}
        \centering
        \includegraphics[width=\linewidth]{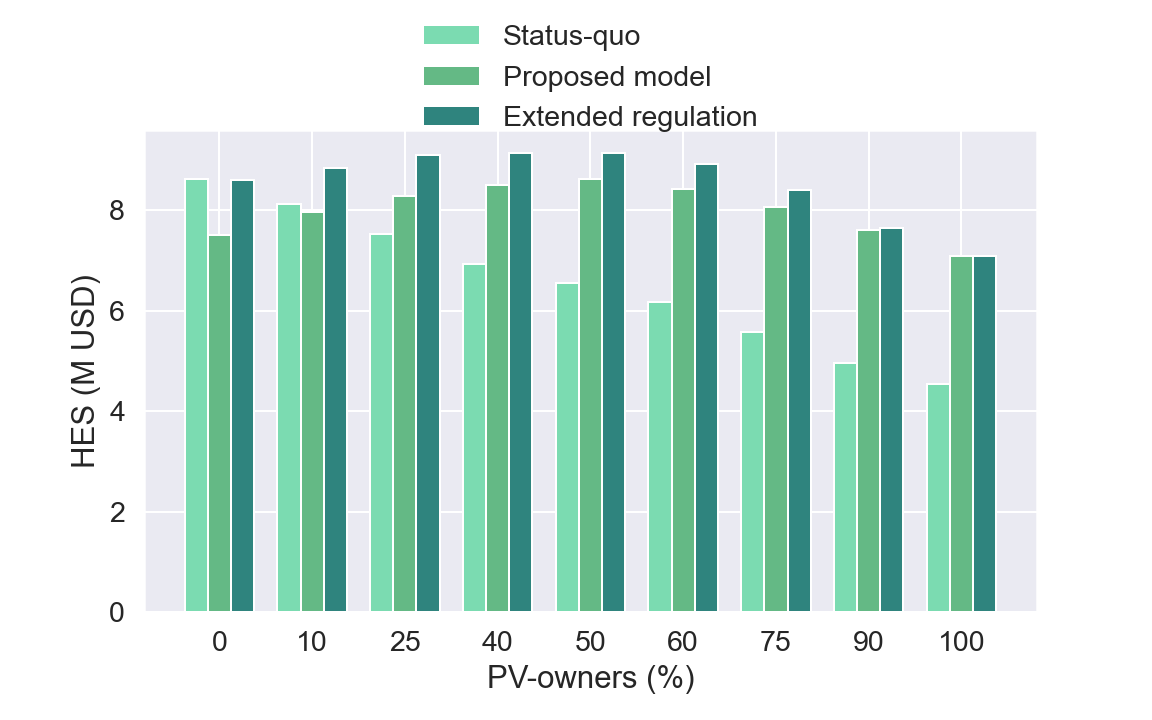}
    \caption{HES under different PV-owner shares}
    \label{fig:HES cons comp pros}
    \end{subfigure}
    \caption{HES for different microgrid cases}
    \label{fig:HES cons comp}
\end{figure}

\section{Conclusion}
\label{sec:conclusion}
In many grid-constrained countries affected by political and economic instability, central governments often fail to meet society’s electricity needs. In response, neighborhood diesel generators have emerged to fill this gap by forming informal, diesel-based microgrids that supply power to nearby households and institutions. With the increasing adoption of household-level photovoltaic (PV) systems, new questions arise regarding the optimal utilization of existing assets to form microgrids capable of delivering more affordable and renewable electricity. Addressing these questions requires not only a techno-economic perspective but also an examination of the market power exercised by incumbent actors, specifically diesel generator companies (DGC), within environments characterized by limited regulatory oversight.

In this paper, two types of contributions are made: 
\begin{itemize}
\item In terms of energy system modeling, we propose a bi-level game theoretical model representing a regulator with limited oversight at the first-level, and a monopolizing diesel generator company at the second. The regulator decides on the maximum electricity price and minimum feed-in-tariff in the microgrid in order to maximize the household economic surplus (HES). The regulator promotes the integration of renewable energy in the microgrid through allowing household PV-owner feed-in and incentivizing DGC-owned PV and battery installation, while accounting for the economic interests of the DGC reflecting its market power. The DGC controls access to the microgrid, and decides on capacity additions, dispatch, and the amount of fed-in capacity from household PV-owners.

\item In terms of policy insights, the model is applied to the case of a real microgrid in a village in Lebanon, relying on local data collected over six months. The following main conclusions are reached: 

Under DGC market power, microgrid-level price and feed-in-tariff caps are generally successful in considerably increasing the HES. The policies that maximize the HES always entail a large portion of household PV-owner fed-in electricity into the microgrid satisfying up to 39\% of the demand in the base case. The introduction of feed-in-tariff regulation and renewable energy in the microgrid yields a significant improvement in HES, particularly through improved electricity affordability and reduced PV-owner generation waste. The proposed model substantially decreases the reliance on diesel-based generation as well. The transition to cheaper, sustainable renewable energy holds both an economic and an environmental importance, especially in politico-economically challenged contexts.

Increasing the DGC's budget leads to substantial gains in HES compared to the status quo. The HES improvements amounting to 18\% at the base budget increase to 41\% under a non-binding budget.

When the household PV-owners' penetration in the microgrid exceeds 10\%, considerable HES improvements compared to a diesel-based microgrid are observed. The improvements peak at a penetration level of around 90\%, where the household PV-owners' generation surplus allows for cheaper electricity, and the demand served by the microgrid is still significant.

The renewable energy penetration in the microgrid reaches 60\% at the base case conditions compared to 0\% in the status quo diesel-based microgrid. This penetration rate approaches 100\% for budgets starting 1 million USD or for PV-owning households constituting 75\%  of all households in the microgrid.

The limited regulator's control over unmet demand entails measurable HES losses under tight budget conditions and low household PV-owner shares in the microgrid. At the status-quo budget and PV-owner share, extending price-cap regulation to also control unmet demand leads to a 10\% increase in HES. In practice, if the regulator cannot directly control unmet demand, it can instead influence system outcomes through targeted state incentives. Conditional low-interest loans for DGCs, linked to renewable energy investments, effectively increase their available investment budget and reduce unmet demand. Likewise, subsidized low-interest loans for household rooftop PV adoption increase PV penetration within the microgrid, which further contributes to lowering unmet demand.

\end{itemize}

This study assumes a microgrid that operates independently of the national grid. Future research could investigate cases of partial interconnection with the national grid, as well as the inclusion of time-varying electricity prices and feed-in-tariffs over the planning horizon. Additionally, incorporating uncertainty and dynamic decision-making throughout the planning horizon represents an interesting direction for future research. This could involve extending the proposed framework to a stochastic or multi-stage bi-level optimization model, or alternatively to a dynamic game in which the regulator and the DGC sequentially update their decisions as uncertainties, such as technology costs and fuel prices, are progressively resolved.

\section{Acknowledgments}
We thank Prof. Georges Zaccour and Prof. Anne Neumann for providing comments and review. This publication is based on research supported by the Templeton World Charity Foundation, Inc. 
(funder DOI 501100011730) under the grant
\url{https://doi.org/10.54224/32645}.

\section{Declaration of generative AI and AI-assisted technologies in the manuscript preparation process}
During the preparation of this work the authors used OpenAI's ChatGPT to assist with language editing and improving overall readability. After using this tool, the authors reviewed and edited the content as needed and take full responsibility for the content of the published article.

\newpage
\bibliographystyle{unsrturl}
\bibliography{bibliography}

\newpage
\appendix
\section{Nomenclature} \label{nomenclature}
The parameters and sets used in the description of the bi-level framework are shown in this appendix.
\renewcommand{\thetable}{A.\arabic{table}}
\setcounter{table}{0}
\begin{table}[h!]
    \centering
    \small
    \caption{Parameters\\
        $^*$ A negative surplus indicates a net demand}
    \resizebox{\textwidth}{!}{
    \renewcommand{\arraystretch}{1}
    
    \begin{tabular}{l p{340pt} l}
        \hline
        Symbol & Input name & Unit\\
        \hline
        $\alpha$ & Minimum state of charge &\\
        $\beta$& Binary input allowing the use of renewable sources & binary\\
        $\gamma^{DGC}$ & Discount rate of the DGC&\\
        $\gamma^{RE}$ & Discount rate of the regulatory entity&\\
         $\varepsilon$ & Unit penalty of unmet demand& USD/kWh\\
         $\eta$ & Charging and discharging efficiency &\\
         $\theta^{PV}_i$ & Average PV capacity of households& kW\\ 
         $\kappa_{g}$ & Initial installed DGC capacity & kW\\
         $\lambda^{C}_{g, y}$ & Unit capital cost of technology& USD/kW\\
         $\lambda^{OF}_{g}$ & Unit fixed operation cost of technology & USD/kW/$y$\\
         $\lambda^{OV}_{g}$ & Unit variable operation cost of technology & USD/kWh\\
         $\mu_{i,y,d,h}$ & Total electricity demand of households & kW\\
         $\nu_g$ & Lifetime of technology & years\\
         $\nu^0_g$ & Remaining lifetime of installed technology at year 0& years\\
         $\xi$ & Value of lost load & USD/kWh\\
         $\Pi$& Budget available for installing new capacity & USD\\
         $\pi$ & Price of diesel & USD/L\\
         $\rho_j$ & Heat rate of the diesel generator on the $j^{th}$ portion of the heat rate curve & L/kWh\\
         $\sigma_{i,y,d,h}$ & Surplus electricity per household $^*$ & kW\\
         $\tau_y$ & Maximum electricity price set by the ministry& USD/kWh\\
         $\upsilon$ & Minimum required level of satisfied demand & \%\\
         $\Upsilon$ & Planning horizon &\\
         $\phi_{y,d,h}$ & Capacity factor of PV &\\
         $\Omega_{i,y}$ & Total number of households available to the microgrid &\\
         $\omega_d$ & Weight of representative day $d$  &\\
        \hline
        \end{tabular}
        }
        
        \label{tab:In}
\end{table} 
\FloatBarrier
\begin{table}[h!]
        \centering
        \small
        \caption{Derived sets}
        \resizebox{\textwidth}{!}{
        \renewcommand{\arraystretch}{1}
        \begin{tabular}{l p{400pt}}
        \hline
         Symbol    & Set name \\
        \hline
        $\mathcal{D}$ & Representative periods (days) within a year ($\left\{0,1,2\right\}$)\\
        $\mathcal{G}$     & All technologies (diesel generator $DG$, DGC photovoltaic cells $PV$, DGC batteries $B$)\\
        $\mathcal{G}g$ & Non-storage technologies ($DG$, $PV$)\\
         $\mathcal{H}$ & Hours in a representative period ($\left\{0,...,23\right\}$)\\
         $\mathcal{I}$ & Household types (Non-PV-owner household, PV-owner household)\\
         $\mathcal{R}$ & Renewable sources ($PV$, fed-in PV-owner surplus $Fi$)\\
         $\mathcal{S}$ & All decision variables\\
         $\mathcal{Y}$ & Years ($\left\{0,...,14\right\}$)\\
        \hline
        \end{tabular}
        }
        \label{tab:sets}
    \end{table}

\section{Second-level constraints} \label{constraints}
\subsection{Capacity and retirement of technology constraints} \label{gen tech}
Equation \ref{inst cap} keeps track of the installed capacity of the technologies $C_{g,y}$, as a function of the previously installed $C_{g,y-1}$, the added $A_{g,y}$, and the retired capacities $Ret_{g,y}$, where the initial capacities are detailed in \ref{init cap}.
    \begin{equation}\label{inst cap}
        C_{g,y}=C_{g, y-1} + A_{g,y} - Ret_{g,y} \quad\forall g, y-1\geq0, \tag{B.1}
    \end{equation}
    \begin{equation}\label{init cap}
        C_{g,0}=\kappa_{g} + A_{g,0} - Ret_{g,0}\quad\forall g \tag{B.2}
    \end{equation}
The retirement of the DGC's capacities is modeled in Equations \ref{ret init} through \ref{ret inst}. When the remaining life of the initial capacity ends, the retired capacity is equal to the initial capacity (\ref{ret init}). Before that, no capacity gets retired (\ref{no ret}). After the initial capacity is retired, the capacity of technology $g$ added $\nu^c_g$ years ago (i.e.: a lifetime ago) is retired (\ref{ret inst}).
    \begin{align}
        Ret_{g,y}&=\kappa_{g} \quad\forall g, y=\nu^0_g \label{ret init}\tag{B.3}\\
        Ret_{g,y}&=0\quad\forall g, y<\nu^0_g \label{no ret}\tag{B.4}\\
        Ret_{g,y}&=A_{g,y-\nu_g}\quad\forall, y>\nu_g \label{ret inst}\tag{B.5}
    \end{align}
    
\subsection{Dispatch constraints} \label{gen disp}
    Equations \ref{DG disp} through \ref{bat in} ensure that the dispatch $D_{g,y,d,h}$ from neither the diesel generator, the DGC's PV nor the batteries charging or discharging exceed their respective capacities.
    \begin{equation}\label{DG disp}
        D_{g=DG,y,d,h}\leq C_{g=DG,y} \quad\forall y,d,h \tag{B.6}
    \end{equation}
    \begin{equation}\label{PV disp}
        D_{g=PV,y,d,h}\leq \phi_{y,d,h} \times C_{g=PV}\;\quad\forall\;y,d,h \tag{B.7}
    \end{equation}
    \begin{equation}\label{bat out}
        B^-_{y,d,h}\leq C_{g=B,y}\quad\forall y,d,h \tag{B.8}
    \end{equation}
    \begin{equation}\label{bat in}
        B^+_{y,d,h}\leq C_{g=B,y} \quad\forall y,d,h \tag{B.9}
    \end{equation}

\subsection{Storage technology constraints} \label{storage tech}
 Equation \ref{SoC} follows the hourly state of charge of the batteries as a function of the previous state of charge, charge, discharge, and efficiency $\eta$, while the initial condition on the state of charge is set in \ref{SoC rep}. Equation \ref{SoC bound} ensures that the state of charge never drops below a given minimum (as a percentage of the energy capacity), and doesn't exceed the energy capacity.
    \begin{equation}\label{SoC}
        SoC_{y,d,h}=SoC_{y,d,h-1}+\eta\times B^+_{y,d,h}- B^-_{y,d,h}/\eta\quad\forall y,d,h-1\geq0, \tag{B.10}
    \end{equation}
    \begin{equation}\label{SoC}
        SoC_{y,d,0}=SoC^0_{y,d}+\eta\times B^+_{y,d,0}- B^-_{y,d,0}/\eta\quad\forall y,d,h-1\geq0, \tag{B.11}
    \end{equation}
    \begin{equation}\label{SoC rep}
        SoC^0_{y,d}=SoC_{y,d,23}+\eta\times B^+_{y,d,23}- B^-_{y,d,23}/\eta\quad\forall y,d \tag{B.12}
    \end{equation}
    \begin{equation}\label{SoC bound}
         \alpha \times 4h \times C_{g=B, y} \leq SoC_{y,d,h} \leq 4h \times C_{g=B, y}\quad\forall y,d,h \tag{B.13}
    \end{equation}

\section{Proof of Lemma 1} \label{lemma proof}
Assuming $NPV^{DGC}(\mathcal{S}^\prime, P, FiT)$ is the NPV for any fixed set of variables $\mathcal{S}^\prime=\mathcal{S}-\{P, FiT\}$, price $P$ and feed-in-tariff $FiT$, and $NPV^{DGC*}(\mathcal{S}^{\prime *}, P, FiT)$ the NPV for the optimal set of variables $\mathcal{S}^{\prime *}=\mathcal{S^*}-\{P, FiT\}$, price $P$ and feed-in-tariff $FiT$.\\ 
With changes of price and feed-in-tariff having no effect on the model's feasibility and parameters, \textit{i.e.,} $P$ and $FiT$ are independent of all constraints, and no input is a function of $P$ and $FiT$, we first prove that the NPV is non-decreasing in $P$.\\
\begin{equation}
    \frac{\partial NPV^{DGC}}{\partial P} = \sum_y \sum_d \omega_d \sum_h \left(B^-_{y,d,h} - B^+_{y,d,h} + \sum_{g \in \mathcal{G}g}D_{g,y,d,h} + \sum_{i\in\mathcal{I}} Fi_{i,y,d,h} \right)\left(\frac{1}{(1+\gamma^{DGC})^y}\right) \tag{C.1}
\end{equation}
Recalling Equation \ref{sup-dem}, where $\forall y,d,h$:
\begin{equation*}
    U_{y,d,h} + B^-_{y,d,h} + \sum_{g\in\mathcal{G}g} D_{g,y,d,h} + \sum_{i\in\mathcal{I}} Fi_{i,y,d,h} =  B^+_{y,d,h} - \sum_{i\in\mathcal{I}}\min\left( 0,  \Omega_{i,y} \times \sigma_{i,y,d,h}\right)
\end{equation*}
\noindent Moreover, equation \ref{min UD h} states that that $\forall y,d,h$
\begin{equation*}
U_{y,d,h} \leq \sum_{i\in\mathcal{I}}-\min\left( 0,  \Omega_{i,y} \times \sigma_{i,y,d,h}\right) \tag{C.2} \label{B.3}
\end{equation*}
\noindent Joining Equations \ref{sup-dem} and \ref{min UD h} implies
\begin{align}
    &&B^-_{y,d,h} + \sum_{g \in \mathcal{G}g}D_{g,y,d,h} + \sum_{i\in\mathcal{I}} Fi_{i,y,d,h} \geq B^+_{y,d,h}&\quad\forall y,d,h \tag{C.3}\\
    &\Leftrightarrow&\sum_d \omega_d \sum_h \left(B^-_{y,d,h} - B^+_{y,d,h} + \sum_{g \in \mathcal{G}g}D_{g,y,d,h} + \sum_{i\in\mathcal{I}} Fi_{i,y,d,h} \right) \geq 0 &\quad\forall y\tag{C.4}\\
    &\Leftrightarrow&\frac{\partial NPV^{DGC}}{\partial P} \geq 0 \tag{C.5}\\
    &\Leftrightarrow& NPV^{DGC}(\mathcal{S}^\prime, P +\epsilon, FiT) \geq NPV^{DGC}(\mathcal{S}^\prime, P, FiT)&\quad\forall  P, FiT, \epsilon\geq0 \tag{C.6}
\end{align}
$\therefore$ $NPV^{DGC}(\mathcal{S}^\prime, P, FiT)$ is non-decreasing in $P$, implying that:
\begin{equation*}
    NPV^{DGC}(\mathcal{S}^{\prime *}, P +\epsilon, FiT) \geq NPV^{DGC*}(\mathcal{S}^{\prime *}, P, Fit) \quad\forall  P, FiT, \epsilon\geq0
\end{equation*}
and, with $NPV^{DGC*}(\mathcal{S}^{\prime **}, P +\epsilon, FiT)$ the maximum NPV, for a re-optimized  $\mathcal{S}^{\prime **}$, price $P +\epsilon, $ and feed-in-tariff $ FiT$, 
\begin{align*}
     && NPV^{DGC*}(\mathcal{S}^{\prime **}, P +\epsilon, FiT) \geq NPV^{DGC}(\mathcal{S}^{\prime *}, P+\epsilon, FiT)&\quad\forall  P, FiT, \epsilon\geq0 \\   
     &\Leftrightarrow& NPV^{DGC*}(\mathcal{S}^{\prime **}, P +\epsilon, FiT) \geq NPV^{DGC*} (\mathcal{S}^{\prime *}, P, FiT)&\quad\forall  P, FiT, \epsilon\geq0
\end{align*}
Then, we prove that $NPV^{DGC}$ is non-increasing in $FiT$. $\forall P, FiT$:
\begin{align}
    &&\frac{\partial NPV^{DGC}}{\partial FiT} = \sum_y\sum_d\omega_d\sum_h\left(-\sum_i Fi_{i,y,d,h}\right)\left(\frac{1}{(1+\gamma^{DGC})^y}\right)& \tag{C.8}\\
    &\Leftrightarrow&\frac{\partial NPV^{DGC}}{\partial FiT} \leq 0&\tag{C.9}\\
    &\Leftrightarrow&NPV^{DGC}(P, FiT+\epsilon) \leq NPV^{DGC}(P, FiT)&\quad\forall \epsilon\geq0\tag{C.10}
\end{align}
$\therefore$ $NPV^{DGC}(\mathcal{S}^\prime, P, FiT)$ is non-increasing in $FiT$, implying that:
\begin{equation*}
    NPV^{DGC*}(\mathcal{S}^{\prime *}, P, FiT +\epsilon) \leq NPV^{DGC}(\mathcal{S}^{\prime *}, P, FiT) \quad\forall  P, FiT, \epsilon\geq0 
\end{equation*}
and, with $NPV^{DGC*}(\mathcal{S}^{\prime **}, P, FiT+\epsilon)$ the maximum NPV, for a re-optimized  $\mathcal{S}^{\prime **}$, price $P, $ and feed-in-tariff $ FiT +\epsilon$, 
\begin{align*}
&&NPV^{DGC}(\mathcal{S}^{\prime *}, P, FiT+\epsilon) \leq NPV^{DGC*}(\mathcal{S}^{\prime **}, P, FiT+\epsilon)&\quad\forall  P, FiT, \epsilon\geq0\\
     &\Leftrightarrow& NPV^{DGC*}(\mathcal{S}^{\prime *}, P, FiT+\epsilon) \leq NPV^{DGC*}(\mathcal{S}^{\prime **}, P, FiT)&\quad\forall  P, FiT, \epsilon\geq0
\end{align*}
\qed

\section{Variable heat rate of diesel generator}\label{VHR}
As first denoted in constraints \ref{HR.1} through \ref{HR.3} model the variable heat rate of the diesel generator. Linearly, they translates to:
    
    \begin{align}
        D_{g=DG, y,d,h} &\leq  0.30 \times C_{g=DG,y} + (1 - b_{1,y,d,h})M \nonumber\\
        R_{y,d,h} &\geq \rho_1 \times D_{g=DG, y,d,h} - (1 - b_{1,y,d,h})M \nonumber\\
        R_{y,d,h} &\leq \rho_1 \times D_{g=DG, y,d,h} + (1 - b_{1,y,d,h})M \label{j=1}\tag{D.1}\\
        D_{g=DG, y,d,h} &\geq  0.30 \times C_{g=DG,y} + (1 - b_{2,y,d,h})M + \epsilon \nonumber\\
        D_{g=DG, y,d,h} &\leq  0.60 \times C_{g=DG,y} + (1 - b_{2,y,d,h})M \nonumber\\
        R_{y,d,h} &\geq \rho_2 \times D_{g=DG, y,d,h} - (1 - b_{2,y,d,h})M \nonumber\\
        R_{y,d,h} &\leq \rho_2 \times D_{g=DG, y,d,h} + (1 - b_{2,y,d,h})M \label{j=2}\tag{D.2}\\
        D_{g=DG, y,d,h} &\geq  0.60 \times C_{g=DG,y} + (b_{1,y,d,h}+b_{2,y,d,h})M + \epsilon \nonumber\\
        R_{y,d,h} &\geq \rho_3 \times D_{g=DG, y,d,h} - (b_{1,y,d,h}+b_{2,y,d,h})M \nonumber\\
        R_{y,d,h} &\leq \rho_3 \times D_{g=DG, y,d,h} + (b_{1,y,d,h}+b_{2,y,d,h})M \label{j=3}\tag{D.3}
    \end{align}

    where $M$ is a large enough number, and $\epsilon$ is a small enough number.\\
    In other words, we model a set of constraints \ref{j=1}, \ref{j=2} and \ref{j=3} for each of cases 1, 2 and 3 respectively, such that only one portion of the heat rate curved is enforced at each time step.

\section{Case study parameters} \label{app:case study}
The costs of the technologies used in the case study are summarized in Table \ref{tab:case study}, noting that the capital expenditure (Capex) of solar PV and batteries decrease along the planning horizon. The operating expenditures (Opex) remain constant.
\renewcommand{\thetable}{E.\arabic{table}}
\setcounter{table}{0}
\begin{table}[h!]
    \centering
    \caption{Case study parameters  \cite{olleik2026planning, ahmad2020distributed, hatton2024global}}
    \begin{tabular}{l l l}
        \hline
        Parameter & Value & Unit\\
        \hline
        PV Capex in year 0 & 964 & USD/kW\\
        PV fixed Opex & 22 & USD/kW/year\\
        PV variable Opex & 0 & USD/kWh\\
        Diesel generator Capex in year 0 & 800 & USD/kW\\
        Diesel generator fixed Opex & 84 & USD/kW/year\\
        Diesel generator non-fuel variable Opex & 0.014 & USD/kWh\\
        Diesel generator fuel-related variable Opex & 0.21 & USD/kWh\\
        Batteries Capex in year 0 & 335 & USD/kW\\
        Batteries fixed Opex & 61 & USD/kW/year\\
        Batteries variable Opex & 0.0006 & USD/kWh\\
        \hline
        \end{tabular}
        \label{tab:case study}
\end{table} 
\end{document}